\documentclass{ws-ijfcs}
\usepackage{enumerate}
\usepackage{url}
\usepackage{ifthen}

\usepackage{cite}
\usepackage{thmtools}
\usepackage{nameref}
\usepackage{hyperref}

\usepackage{amsmath}
\usepackage{amsfonts}
\usepackage{amssymb}

\usepackage{ascmac}
\usepackage{fancybox}

\usepackage{cleveref}
\newtheorem{thm}{Theorem}
\crefname{thm}{Thm.}{Thms.}

\newtheorem{prop}[thm]{Proposition}
\crefname{prop}{Prop.}{Props.}

\newtheorem{lem}[thm]{Lemma}
\crefname{lem}{Lem.}{Lems.}

\newtheorem{cor}[thm]{Corollary}
\crefname{cor}{Cor.}{Cors.}

\newtheorem{claim}[thm]{Claim}
\crefname{claim}{claim.}{claims.}

\newtheorem{defi}[thm]{Definition}
\crefname{defi}{Def.}{Defs.}

\newtheorem{ex}[thm]{Example}

\newtheorem{rem}[thm]{Remark}

\crefname{figure}{Fig.}{Figs.}

\crefname{section}{Sect.}{Sects.}
\crefname{appendix}{Appendix}{}

\usepackage{tikz}
\usetikzlibrary{shapes, arrows, arrows.spaced, arrows.meta, chains,positioning, cd,
    decorations, decorations.pathmorphing, decorations.markings, decorations.pathreplacing, automata, backgrounds, petri, matrix, fit, calc, graphs, quotes, bending}

\usepackage{ebproof}
\ebproofnewstyle{small}{separation = .5em, label separation= .1em, right label template= \footnotesize \inserttext}
\usepackage{stackengine}

\usepackage{empheq}

\usepackage{textcomp}
\usepackage{stmaryrd} %
\usepackage{stackengine} %

\usepackage{tablefootnote}
\usepackage{multirow}
\usepackage{diagbox}

\usepackage{pict2e,picture}
\usepackage{xparse}
\usepackage{accents}
\usepackage{xcolor}
\definecolor[named]{ACMBlue}{cmyk}{1,0.1,0,0.1}
\definecolor[named]{ACMYellow}{cmyk}{0,0.16,1,0}
\definecolor[named]{ACMOrange}{cmyk}{0,0.42,1,0.01}
\definecolor[named]{ACMRed}{cmyk}{0,0.90,0.86,0}
\definecolor[named]{ACMLightBlue}{cmyk}{0.49,0.01,0,0}
\definecolor[named]{ACMGreen}{cmyk}{0.20,0,1,0.19}
\definecolor[named]{ACMPurple}{cmyk}{0.55,1,0,0.15}
\definecolor[named]{ACMDarkBlue}{cmyk}{1,0.58,0,0.21}
\hypersetup{colorlinks,
    linkcolor=ACMPurple,
    citecolor=ACMPurple,
    urlcolor=ACMDarkBlue,
    filecolor=ACMDarkBlue}

\usepackage[normalem]{ulem}

\usepackage{enumitem}
\makeatletter
\let\orgdescriptionlabel\descriptionlabel
\renewcommand*{\descriptionlabel}[1]{%
    \let\orglabel\label
    \let\label\@gobble
    \phantomsection
    \edef\@currentlabel{#1\unskip}%
    \let\label\orglabel
    \orgdescriptionlabel{#1}%
}
\makeatother

\usepackage{etoolbox}%
\patchcmd{\footnotemark}{\stepcounter{footnote}}{\refstepcounter{footnote}}{}{}

\usepackage[notion,paper]{knowledge}

\usepackage{snotez}

\newcommand{\set}[1]{\{ #1 \}}
\newcommand{\tuple}[1]{\langle #1 \rangle}
\newcommand{\const}[1]{\mathsf{#1}}
\newcommand{\bl}{\_}
\newcommand{\defeq}{\mathrel{\ensurestackMath{\stackon[1pt]{=}{\scriptscriptstyle\Delta}}}}
\newcommand{\defiff}{\mathrel{\ensurestackMath{\stackon[1pt]{\Leftrightarrow}{\scriptscriptstyle\Delta}}}}

\newcommand{\nat}{\mathbb{N}}
\newcommand{\card}{\mathop{\#}}
\newcommand{\range}[1]{[#1]}

\NewDocumentCommand\word{O{1}}{%
    \ifcase#1
        undefined
    \or w
    \or v
    \or u
    \else undefined

    \fi
}
\newcommand{\len}[1]{\|#1\|}
\NewDocumentCommand\la{O{1}}{%
    \ifcase#1
        undefined
    \or L
    \or K
    \else undefined

    \fi
}

\newcommand{\vsig}{\mathbf{V}}

\NewDocumentCommand\term{O{1}}{%
    \ifcase#1
        undefined
    \or t
    \or s
    \or u
    \else undefined

    \fi
}

\NewDocumentCommand\alg{O{1}}{%
    \ifcase#1
        undefined
    \or \mathcal{A}
    \or \mathcal{B}
    \else undefined

    \fi
}

\NewDocumentCommand\algclass{O{1}}{%
    \ifcase#1
        undefined
    \or \mathcal{C}
    \else undefined

    \fi
}

\newcommand{\sig}{S}
\newcommand{\union}{\mathbin{+}}
\newcommand{\compo}{\mathbin{;}}
\newcommand{\kstar}{*}
\newcommand{\emp}{\const{0}}
\newcommand{\id}{\const{1}}
\newcommand{\eps}{\varepsilon}
\newcommand{\compl}{-}

\newcommand{\domain}[1]{|#1|}
\newcommand{\val}{\mathfrak{v}}
\newcommand{\ljump}[1]{[ #1 ]}
\newcommand{\lang}{\mathsf{lang}}
\newcommand{\LANG}{\mathsf{LANG}}

\newcommand{\KAtermclass}{\mathrm{KA}}
\newcommand{\com}[1]{\overline{#1}}

\newcommand{\langlen}{\mathop{\mathrm{l}}} %
\knowledge{ignore}
| letter
| letters

\knowledge{ignore}
| word
| words

\knowledge{ignore}
| length

\knowledge{ignore}
| language
| languages

\knowledge{ignore}
| composition

\knowledge{ignore}
| Kleene star

\knowledge{notion}
| $\sig$-algebra
| $\sig$-algebras

\knowledge{notion}
| term
| terms

\knowledge{notion}
| equation
| equations

\knowledge{notion}
| inequation
| inequations

\knowledge{notion}
| equational theory
| equational theories
| equational theory of $\algclass$

\knowledge{notion}
| equational theory w.r.t.\ languages

\knowledge{notion}
| valuation
| valuations

\knowledge{notion}
| language model
| language models

\knowledge{notion}
| star-free

\knowledge{notion}
| words-to-letters valuation
| words-to-letters valuations
| Words-to-letters valuations

\knowledge{notion}
| letters-to-letters valuation
| letters-to-letters valuations

\knowledge{ignore}
| variable
| variables

\knowledge{ignore}
| constant
| constants

\knowledge{ignore}
| supremum length

\knowledge{ignore}
| variable complements

\knowledge{ignore}
| constant complements

\knowledge{notion}
| positive

\knowledge{notion}
| negative

 \usepackage{environ}
\newboolean{isdraft}

\setboolean{isdraft}{true}

\ifthenelse{\boolean{isdraft}}{
  \setlength{\marginparwidth}{7em}
  \setlength{\marginparsep}{0.1em}
  \NewEnviron{commentyn}{\sidenote{\textcolor{green}{YN: \BODY}}}
}{
  \NewEnviron{commentyn}{}
}

\begin{document}

\markboth{Authors' Names}
{Instructions for Typing Manuscripts $($Paper's Title\/$)$}

\catchline{}{}{}{}{}
\title{Words-to-Letters Valuations for Language Kleene Algebras with Variable and Constant Complements}
\author{Yoshiki Nakamura}
\address{Institute of Science Tokyo, Japan\\
    \email{nakamura.yoshiki.ny@gmail.com}}
\author{Ryoma Sin'ya}
\address{Akita University, Japan\\
    \email{ryoma@math.akita-u.ac.jp}}

\maketitle

\begin{history}
    \received{(Day Month Year)}
    \accepted{(Day Month Year)}
    \comby{(xxxxxxxxxx)}
\end{history}

\begin{abstract}
We investigate the equational theory for Kleene algebra terms with \emph{variable complements} and \emph{constant complements}---(language) complement where it applies only to variables or constants---w.r.t.\ languages.
While the equational theory w.r.t.\ languages coincides with the language equivalence (under the standard language valuation) for Kleene algebra terms, this coincidence is broken if we extend the terms with complements.
In this paper, we prove the decidability of some fragments of the equational theory: the universality problem is coNP-complete, and the inequational theory $t \le s$ is coNP-complete when $t$ does not contain Kleene-star.
To this end, we introduce \emph{words-to-letters valuations};
they are sufficient valuations for the equational theory and ease us in investigating the equational theory w.r.t.\ languages.
Additionally, we show a completeness theorem of the equational theory for words with variable complements and the non-empty constant.
 \end{abstract}


\section{Introduction}
Kleene algebra (KA) \cite{kleeneRepresentationEventsNerve1951,conwayRegularAlgebraFinite1971} is an algebraic system for regular expressions consisting of union ($\union$), composition ($\compo$), Kleene-star ($\bl^{\kstar}$), empty ($\emp$), and identity ($\id$).
In this paper, we consider KAs \emph{w.r.t.\ languages} (a.k.a., \kl{language models} of KAs, language KAs).
Interestingly, the \kl[equational theory w.r.t.\ languages]{equational theory of KAs w.r.t.\ languages} coincides with the language equivalence under the standard language valuation (see also, e.g., \cite{andrekaEquationalTheoryKleene2011,pousCompletenessTheoremsKleene2022}):
for all KA \kl{terms} (i.e., regular expressions) $\term[1], \term[2]$, we have
\begin{align*}
  \label{equation: LANG = lang}\LANG \models \term[1] = \term[2] \quad \Leftrightarrow \quad \ljump{\term[1]} = \ljump{\term[2]} \tag{$\dagger$}.
\end{align*}
Here, we write $\LANG \models \term[1] = \term[2]$ if the equation $\term[1] = \term[2]$ holds for all \kl{language models} (i.e., each \kl{variable} $x$ maps to not only the singleton language $\set{x}$ but also any \kl{languages});
we write $\ljump{\term[3]}$ for the \kl{language} of a regular expression $\term[3]$ (i.e., each variable $x$ maps to the singleton \kl{language} $\set{x}$).
Since the valuation $\ljump{\bl}$ is an instance of valuations in $\LANG$, the direction $\Rightarrow$ is trivial (this direction always holds even if we extend KA \kl{terms} with some extra operators).
The direction $\Leftarrow$ is a consequence of the completeness of KAs (see \Cref{theorem: LANG and lang} for an alternative proof not relying on the completeness of KAs).
However, the direction $\Leftarrow$ fails when we extend KA \kl{terms} with some extra operators;
thus, the \kl{equational theory w.r.t.\ languages} does not coincide with the language equivalence (see below and \Cref{remark: LANG and lang} for complements).
The \kl{equational theory w.r.t.\ languages} of KAs with some operators was studied,
e.g., with reverse \cite{bloomNotesEquationalTheories1995},
with tests \cite{kozenKleeneAlgebraTests1996} (where languages are of guarded strings, not words),
with intersection ($\cap$) \cite{andrekaEquationalTheoryKleene2011},
with universality ($\top$) \cite{zhangIncorrectnessLogicKleene2022,pousCompletenessTheoremsKleene2022},
and combinations of some of them \cite{brunetReversibleKleeneLattices2017,brunetCompleteAxiomatisationFragment2020}.

Nevertheless, to the best of authors' knowledge, \kl{variable complements} (and even complements) w.r.t. languages has not yet been investigated,
while those w.r.t.\ binary relations were studied, e.g.,\ in \cite{ngRelationAlgebrasTransitive1984} (for complements, cf.\ Tarski's calculus of relations \cite{tarskiCalculusRelations1941}) and \cite{nakamuraExistentialCalculiRelations2023} (for \kl{variable complements}).

In this paper, we investigate the \kl{equational theory} for KA \kl{terms} with \intro*\kl{variable complements} ($\com{x}$) ($x$ denotes a \kl{variable}) and \intro*\kl{constant complements} ($\com{\id}$)---(language) complement, where it applies only to \kl{variables} or \kl{constants}---w.r.t.\ languages; we denote by $\KAtermclass_{\set{\com{x}, \com{\id}}}$ the \kl{terms}.
For $\KAtermclass_{\set{\com{x}, \com{\id}}}$ \kl{terms}, (\ref{equation: LANG = lang}) fails.
The following is a counter-example:
\begin{align*}
  \LANG & \not\models \com{x} = \com{x} \compo \com{x}, & \ljump{\com{x}} & = \ljump{\com{x} \compo \com{x}}.
\end{align*}
($\LANG \not\models \com{x} = \com{x} \compo \com{x}$ is shown by a \kl{valuation} such that $\com{x}$ maps to the language $\set{x}$.
On the other hand, when $\vsig$ denotes the alphabet, $\ljump{\com{x}} = \vsig^* \setminus \set{x} = \ljump{\com{x} \compo \com{x}}$.)
As the example above (see also \Cref{remark: LANG and lang}, for more examples) shows, for $\KAtermclass_{\set{\com{x}, \com{\id}}}$ \kl{terms},
the \kl{equational theory w.r.t.\ languages} significantly differs from the language equivalence under the standard language valuation.
While the language equivalence problem for $\KAtermclass_{\set{\com{x}, \com{\id}}}$ is decidable in PSPACE by a standard automata construction \cite{mcnaughtonRegularExpressionsState1960,thompsonProgrammingTechniquesRegular1968} (and hence, PSPACE-complete \cite{meyerEquivalenceProblemRegular1972, Meyer1973, Hunt1976}),
it remains whether the \kl{equational theory w.r.t.\ languages} is decidable.\footnote{The PSPACE decidability for $\KAtermclass_{\set{\com{x}, \com{\id}}}$ \kl{terms} are recently presented by the first author \cite{nakamuraFiniteRelationalSemantics2025}, by combining the idea of \kl{words-to-letters valuations} and the techniques for relational models in \cite{nakamuraExistentialCalculiRelations2023}.}

We prove the decidability and complexity of some fragments of the \kl{equational theory w.r.t.\ languages} for $\KAtermclass_{\set{\com{x}, \com{\id}}}$ \kl{terms}:
the universality problem is coNP-complete (\Cref{corollary: universality decidable}), and the inequational theory $\term[1] \le \term[2]$ is coNP-complete when $\term[1]$ does not contain Kleene-star (\Cref{corollary: star-free}).
To this end, we introduce \emph{\kl{words-to-letters valuations}}.
\kl{Words-to-letters valuations} are sufficient for the \kl{equational theory w.r.t.\ languages} for $\KAtermclass_{\set{\com{x}, \com{\id}}}$ \kl{terms} (\Cref{corollary: word witness}):
Given \kl{terms} $\term[1], \term[2]$, if some \kl{valuation} refutes $\term[1] = \term[2]$, then some \kl{words-to-letters valuation} refutes $\term[1] = \term[2]$.
This property eases us in investigating the \kl{equational theory w.r.t.\ languages}.

Additionally, we show a completeness theorem of the \kl{equational theory} of $\LANG_{\alpha}$ for the word fragment of $\KAtermclass_{\set{\com{x}, \com{\id}}}$ \kl{terms}
where $\LANG_{\alpha}$ denotes language models over sets of cardinality at most $\alpha$.
A limitation of \kl{words-to-letters valuations} is that the number of \kl{letters} is not bounded, so they may not be compatible with $\LANG_{n}$ where $n$ is a natural number.
For that reason, we give other \kl{valuations} for separating \kl{words} with complement.

\subsection*{Difference with the conference version}
This paper is an extended and revised version of the paper presented at the 16th International Conference on Automata and Formal Languages (AFL 2023) \cite{nakamuraWordstoLettersValuationsLanguage2023}.
The three main differences from the conference version are as follows.
\begin{enumerate}
  \item \label{contribution: 1} We extend \kl{terms} with the complement of the identity constant ($\com{\id}$).\footnote{The universal constant (the complement of the empty constant) $\top$ can be expressed by using the complement of the identity constant (or variables) as $\top = \id \cup \com{\id}$. Thus, we omit $\top$.}
  We can naturally extend the complexity results in \cite{nakamuraWordstoLettersValuationsLanguage2023} while we should carefully treat the empty \kl{word} and non-empty \kl{words} (e.g., \Cref{section: variable inclusion}).
  \item \label{contribution: 2} We strengthen the results of \cite[Thm.\ 35 and 36]{nakamuraWordstoLettersValuationsLanguage2023} from one variable \kl{words} with variable complements to many variables \kl{words} with variable complements and the constant $\com{\id}$ (\Cref{theorem: completeness word LANG0,theorem: completeness word LANG1,theorem: completeness word LANG2,corollary: completeness word LANG2}).
  We had left this problem (more precisely, \Cref{corollary: completeness word LANG2}) open in the conference version \cite{nakamuraWordstoLettersValuationsLanguage2023}.
  While the \kl{equational theory} for \kl{words} with variable complements coincides with the \kl{word} equivalence \cite[Thm.\ 36]{nakamuraWordstoLettersValuationsLanguage2023},
  that for \kl{words} with variable complements and $\com{\id}$ contains non-trivial \kl{equations}, e.g., $\com{\id} x \com{x} \com{\id} = \com{\id} \com{x} x \com{\id}$ (\Cref{example: LANG2}).
  \item \label{contribution: 3} \Cref{section: hierarchy} is new.
  We show that for $\KAtermclass$ with full complement, the \kl{equational theory} of $\LANG_{n}$ does not coincide with \kl[equational theory]{that} of $\LANG_{n+1}$ for each $n \in \nat$.
  For $\KAtermclass$, they are the same \kl{equational theory} for $n \ge 2$.
  We leave open for $\KAtermclass$ with \kl{variable complements} and \kl{constant complements}.
\end{enumerate}
Additionally, some proofs (\Cref{lemma: identity val,lemma: ell abstraction,lemma: ell abstraction gen}) are simplified without induction, based on the alternative semantics using \kl{word} \kl{languages} (\Cref{lemma: lang val}).

\subsection*{Outline}
In \Cref{section: preliminaries}, we briefly give basic definitions, including the syntax and semantics of $\KAtermclass_{\set{\com{x}, \com{\id}}}$ \kl{terms}.
In \Cref{section: identity,section: words_to_letters}, we consider fragments of the \kl{equational theory w.r.t.\ languages} for $\KAtermclass_{\set{\com{x}, \com{\id}}}$ \kl{terms}, step-by-step.
In \Cref{section: identity}, we consider the identity inclusion problem ($\LANG \models \id \le \term$?).
This problem is relatively easy but contains the coNP-hardness result (\Cref{corollary: identity decidable}).
In \Cref{section: words_to_letters}, we consider the variable inclusion problem ($\LANG \models x \le \term$?) and the word inclusion problem ($\LANG \models \word \le \term$?).
For them, we introduce \kl{words-to-letters valuations} (\Cref{definition: val for words to letters}).
Consequently, the \kl[equational theory]{inequational theory} $\term[1] \le \term[2]$ is coNP-complete when $\term[1]$ does not contain Kleene-star (\Cref{corollary: star-free}), including the universality problem ($\LANG \models \top \le \term$?).
Additionally, we show the words-to-letters valuation property (\Cref{corollary: word witness}) for the \kl{equational theory w.r.t.\ languages} for $\KAtermclass_{\set{\com{x}, \com{\id}}}$ \kl{terms}.
In \Cref{section: hierarchy}, we consider the hierarchy of $\LANG_{n}$.
We show that the hierarchy is infinite for $\KAtermclass$ \kl{terms} with full complement, while the hierarchy is collapsed for $\KAtermclass$ \kl{terms}.
In \Cref{section: completeness word}, we consider the \kl{equational theory} for \kl{words} with \kl{variable complements} and the constant $\com{\id}$ and show a completeness theorem (\Cref{theorem: completeness word LANG2}).
\Cref{section: conclusion} concludes this paper.

\section{Preliminaries}\label{section: preliminaries}
We write $\nat$ for the set of non-negative integers.
For $\ell, r \in \nat$, we write $\range{\ell, r}$ for the set $\set{i \in \nat \mid \ell \le i \le r}$.
For a set $X$, we write $\card X$ for the cardinality of $X$ and $\wp(X)$ for the power set of $X$.

For a set $X$ (of \intro*\kl{letters}) and $n \in \nat$, 
we write $X^{\kstar}$ for the set of \intro*\kl{words} over $X$ (finite sequences of elements of $X$).
We write $\len{\word}$ for the \intro*\kl{length} of a \kl{word} $\word$.
We write $X^{n}$ for the set $\set{\word \in X^* \mid \len{\word} = n}$ and write $X^{+}$ for the set $\set{\word \in X^* \mid 1 \le \len{\word}}$.
We write $\eps$ for the empty word.
We write $\word[1] \word[2]$ for the concatenation of \kl{words} $\word[1]$ and $\word[2]$.
A \intro*\kl{language} over $X$ is a subset of $X^{\kstar}$.
We use $\word[1], \word[2]$ to denote \kl{words}
and use $\la[1], \la[2]$ to denote \kl{languages}, respectively.
For \kl{languages} $\la[1], \la[2] \subseteq X^{\kstar}$, the \intro*\kl{composition} $\la[1] \compo \la[2]$ and the \intro*\kl{Kleene star} $\la[1]^{\kstar}$ is defined by:
\begin{align*}
    \la[1] \compo \la[2] & \ \defeq\  \set{\word[1] \word[2] \mid \word[1] \in \la[1] \ \land \ \word[1] \in \la[2]},                        \\
    \la[1]^{\kstar}      & \ \defeq\  \set{\word[1]_0 \dots \word[1]_{n-1} \mid \exists n \in \nat, \forall i < n,\  \word[1]_i \in \la[1]}.
\end{align*}

\subsection{Syntax: terms of KA with complement}
We consider \kl{terms} over the signature $\sig \defeq \set{\id_{(0)}, \emp_{(0)}, \compo_{(2)}, \union_{(2)}, {\bl^{\kstar}}_{(1)}, {\bl^{\compl}}_{(1)}}$ (where complement only applies to variables or constants in the most part).
Let $\vsig$ be a countably finite set of \intro*\kl{variables}.
For a \kl{term} $\term$ over $\sig$,
let $\com{\term}$ be $\term[2]$ if $\term = \term[2]^{-}$ for some $\term[2]$ and be $\term^{-}$ otherwise.
We use the following abbreviations:
\begin{align*}
    \top &\ \defeq\  \emp^{-}, &\term[1] \cap \term[2] &\ \defeq\  (\term[1]^{-} \union \term[2]^{-})^{-}.
\end{align*}
For $X \subseteq \set{\com{x}, \com{\id}, -}$,
let $\KAtermclass_{X}$ be the minimal set $A$ of \kl{terms} over $\sig$ satisfying the following:
\begin{align*}
    \begin{prooftree}
        \hypo{y \in \vsig}
        \infer1{y \in A}
    \end{prooftree}
    \qquad
    \begin{prooftree}
        \hypo{\mathstrut}
        \infer1{\id \in A}
    \end{prooftree}
    \qquad
    \begin{prooftree}
        \hypo{\mathstrut}
        \infer1{\emp \in A}
    \end{prooftree}
    \qquad
    \begin{prooftree}
        \hypo{\term[1] \in A}
        \hypo{\term[2] \in A}
        \infer2{\term[1] \compo \term[2] \in A}
    \end{prooftree}
    \qquad
    \begin{prooftree}
        \hypo{\term[1] \in A}
        \hypo{\term[2] \in A}
        \infer2{\term[1] \union \term[2] \in A}
    \end{prooftree}
    \\
    \begin{prooftree}
        \hypo{\term[1] \in A}
        \infer1{\term[1]^{\kstar} \in A}
    \end{prooftree}
    \qquad
    \begin{prooftree}
        \hypo{\com{x} \in X}
        \hypo{y \in \vsig}
        \infer2{\com{y} \in A}
    \end{prooftree}
    \qquad
    \begin{prooftree}
        \hypo{\com{\id} \in X}
        \infer1{\com{\id} \in A}
    \end{prooftree}
    \qquad
    \begin{prooftree}
        \hypo{- \in X}
        \hypo{\term[1] \in A}
        \infer2{\term[1]^{-} \in A}
    \end{prooftree}.
\end{align*}
We use parentheses in ambiguous situations.
We often abbreviate $\term[1] \compo \term[2]$ to $\term[1] \term[2]$.
We write $\sum_{i = 1}^{n} \term[1]_i$ for the \kl{term} $\emp \union \term[1]_1 \union \dots \union \term[1]_n$.
In the sequel, we mainly consider about $\KAtermclass_{\set{\com{x}, \com{\id}}}$.

An \intro*\kl{equation} $\term[1] = \term[2]$ is a pair of \kl{terms}.
An \intro*\kl{inequation} $\term[1] \le \term[2]$ abbreviates the \kl{equation} $\term[1] \union \term[2] = \term[2]$.

\subsection{Semantics: language models}
An \intro*\kl{$\sig$-algebra} $\alg$ is a tuple $\tuple{\domain{\alg}, \set{f^{\alg}}_{f_{(k)} \in \sig}}$, where $\domain{\alg}$ is a non-empty set and $f^{\alg} \colon \domain{\alg}^{k} \to \domain{\alg}$ is a $k$-ary map for each $f_{(k)} \in \sig$.
A \intro*\kl{valuation} $\val$ of an \kl{$\sig$-algebra} $\alg$ is a map $\val \colon \vsig \to \domain{\alg}$.
For a \kl{valuation} $\val$, we write $\hat{\val} \colon \KAtermclass_{\set{-}} \to \domain{\alg}$ for the unique homomorphism extending $\val$.
We use $\algclass$ to denote a class of \kl{valuations}.
For a \kl{valuation} $\val$ and a class $\algclass$ of \kl{valuations},
we write:
\begin{align*}
    \val \models \term[1] = \term[2] &\ \defiff\  \hat{\val}(\term[1]) = \hat{\val}(\term[2]), &
    \algclass \models \term[1] = \term[2] &\ \defiff\  \forall \val \in \algclass,  \val \models \term[1] = \term[2].
\end{align*}
The \intro*\kl{equational theory of $\algclass$} is the set of all \kl{equations} $\term[1] = \term[2]$ such that $\algclass \models \term[1] = \term[2]$.

The \intro*\kl{language model} $\alg$ over a set $X$, written $\lang_{X}$, is an \kl{$\sig$-algebra} such that $\domain{\alg} = \wp(X^{\kstar})$ and
for all $\la[1], \la[2] \subseteq X^{\kstar}$,
\begin{flalign*}
    \phantom{x^{\alg}} & \phantom{= \set{x}} & \id^{\alg}                 & = \set{\eps},                   &
                    &                     & \la[1] \compo^{\alg} \la[2] & = \la[1] \compo \la[2],         &
                    &                     & \la[1]^{\kstar^{\alg}}      & = \la[1]^{\kstar},                \\
                    &                     & \emp^{\alg}                 & = \emptyset,                   &
                    &                     & \la[1] \union^{\alg} \la[2] & = \la[1] \cup \la[2],         &
                    &                     & \la[1]^{\compl^{\alg}}      & = X^{\kstar} \setminus \la[1].
\end{flalign*}
We write $\LANG_{X}$ for the class of all \kl{valuations} of $\lang_{X}$ and
we write $\LANG$ for $\bigcup_{X} \LANG_{X}$ and write $\LANG_{\alpha}$ for $\bigcup_{X; \# X \le \alpha} \LANG_{X}$.
The \intro*\kl{equational theory w.r.t.\ languages} denotes \kl[equational theory]{that} of $\LANG$.

The \kl{language} $\ljump{\term} \subseteq \vsig^{\kstar}$ of a $\KAtermclass_{\set{-}}$ \kl{term} $\term$ is the \kl{language} $\hat{\val}_{\mathrm{st}}(\term)$ where $\val_{\mathrm{st}}$ is the \kl{valuation} on the \kl{language model} over the set $\vsig$ defined by $\val_{\mathrm{st}}(x) = \set{x}$ for $x \in \vsig$.
Since $\val_{\mathrm{st}} \in \LANG$, we have that for all $\term[1], \term[2]$,
\begin{align*}
    \label{equation: LANG to lang} \LANG \models \term[1] = \term[2] \quad \Rightarrow \quad \ljump{\term[1]} = \ljump{\term[2]}. \tag{$\ddagger$}
\end{align*}
\begin{rem}\label{remark: LANG and lang}
    The converse direction of (\ref{equation: LANG to lang}) fails.
    The following are examples where $x, y \in \vsig$ are distinct \kl{variables} and $\word$ is a \kl{word} over $\vsig$ s.t.\ $\word \neq x$:
    \begin{align}
        \LANG &\not\models y \le \com{x}, & \ljump{y} &\subseteq \ljump{\com{x}},\\
        \LANG &\not\models \word \le \com{x}, & \ljump{\word} &\subseteq \ljump{\com{x}},\\
        \LANG &\not\models y \le \com{\id}, & \ljump{y} &\subseteq \ljump{\com{\id}}, \\
        \LANG &\not\models \com{x} = \com{x} \compo \com{x}, & \ljump{\com{x}}  &= \ljump{\com{x} \compo \com{x}},\\
        \LANG &\not\models \top = \com{x} \compo \com{y},       & \ljump{\top}  &= \ljump{\com{x} \compo \com{y}},\\
        \LANG &\not\models \top = \com{x} \union \com{y},       & \ljump{\top}  &= \ljump{\com{x} \union \com{y}}.
    \end{align}
    (Note that $\term[1] \le \term[2]$ denotes the \kl{equation} $\term[1] \union \term[2] = \term[2]$.)
    For example, for $\LANG \not\models y \le \com{x}$, consider a \kl{valuation} $\val \in \LANG_{\vsig}$ s.t.\
    $\val(x) = \vsig^* \setminus \set{x}$ and $\val(y) = \set{y}$; then we have $y \in \hat{\val}(y) \setminus \hat{\val}(\com{x})$.
    Similarly to the other ``$\LANG \not\models$'', they are shown by considering \kl{valuations} mapping complemented variable to a singleton \kl{language}.
\end{rem}
As the examples above show, for $\KAtermclass$ \kl{terms} with \kl{variable complements} or \kl{constant complements}, the \kl{equational theory w.r.t.\ languages} ($\LANG \models \term[1] = \term[2]$?) significantly differs from the language equivalence problem ($\ljump{\term[1]} = \ljump{\term[2]}$?).
In the sequel, we focus on the \kl{equational theory w.r.t.\ languages} and investigate its fragments.

\subsection{Alternative semantics using (extended) word languages}\label{section: term to lang}
For $\KAtermclass_{\set{\com{x}, \com{\id}}}$ \kl{terms}, we can give an alternative semantics of $\LANG$ using (extended) word \kl{languages}.
The semantics (\Cref{lemma: lang val}) is useful as we can decompose $\KAtermclass_{\set{\com{x}, \com{\id}}}$ \kl{terms} into sets of \kl{words}.

Let $\tilde{\vsig} \defeq \set{x, \com{x} \mid x \in \vsig}$ and let $\tilde{\vsig}_{\com{\id}} \defeq \tilde{\vsig} \cup \set{\com{\id}}$.
For a $\KAtermclass_{\set{\com{x}, \com{\id}}}$ \kl{term} $\term$, we write $\ljump{\term}_{\tilde{\vsig}_{\com{\id}}}$ for the \kl{language} of $\term$ where $\term$ is viewed as the regular expression over $\tilde{\vsig}_{\com{\id}}$.
Each \kl{word} over $\tilde{\vsig}_{\com{\id}}$ is viewed as a $\KAtermclass_{\set{\com{x}, \com{\id}}}$ \kl{term} consisting of composition ($\compo$), variables ($x$), complemented variables ($\com{x}$), and the non-empty constant ($\com{\id}$).
Note that $\ljump{\com{x}}_{\tilde{\vsig}_{\com{\id}}} = \set{\com{x}}$, cf.\ $\ljump{\com{x}} = \vsig^* \setminus \set{x}$.
For a \kl{valuation} $\val \in \LANG$ and a \kl{language} $\la$ over $\tilde{\vsig}_{\com{\id}}$, we define:
\[\hat{\val}(\la) \ \defeq\  \bigcup_{\word \in \la} \hat{\val}(\word). \]
By the distributive law of $\compo$ w.r.t.\ $\union$, for all \kl{valuations} $\val \in \LANG$, we have:
\begin{align*}
    \hat{\val}(\la[1] \union \la[2]) & = \hat{\val}(\la[1]) \cup \hat{\val}(\la[2]), &
    \hat{\val}(\la[1] \compo \la[2]) & = \hat{\val}(\la[1]) \compo \hat{\val}(\la[2]), &
    \hat{\val}(\la[1]^{\kstar})      & = \hat{\val}(\la[1])^{\kstar}.
\end{align*}
Thus, we can decompose each $\KAtermclass_{\set{\com{x}, \com{\id}}}$ \kl{term} $\term$ to the \kl{language} $\ljump{\term}_{\tilde{\vsig}_{\com{\id}}}$ as follows.
\begin{lem}\label{lemma: lang val}
    Let $\val \in \LANG$.
    For all $\KAtermclass_{\set{\com{x}, \com{\id}}}$ \kl{terms} $\term$, we have: $\hat{\val}(\term) = \hat{\val}(\ljump{\term}_{\tilde{\vsig}_{\com{\id}}})$.
\end{lem}
\begin{proof}
    By easy induction on $\term$ using the equations above.
    Case $\term = x, \com{x}, \id, \com{\id}$:
    Clear, by $\ljump{\term}_{\tilde{\vsig}_{\com{\id}}} = \set{\term}$.
    Case $\term = \emp$:
    By $\hat{\val}(\emp) = \emptyset = \hat{\val}(\ljump{\emp}_{\tilde{\vsig}_{\com{\id}}})$.
    Case $\term = \term[2] \union \term[3]$, Case $\term = \term[2] \compo \term[3]$, Case $\term = \term[2]^{\kstar}$:
    By IH with the equations above.
    For example, when $\term = \term[2] \compo \term[3]$, we have:
    \begin{align*}
        \hat{\val}(\term[2] \compo \term[3])  = \hat{\val}(\term[2]) \compo \hat{\val}(\term[3]) & = \hat{\val}(\ljump{\term[2]}_{\tilde{\vsig}_{\com{\id}}}) \compo \hat{\val}(\ljump{\term[3]}_{\tilde{\vsig}_{\com{\id}}}) \tag{IH}                                    \\
                                                                                                 & = \hat{\val}(\ljump{\term[2]}_{\tilde{\vsig}_{\com{\id}}} \compo \ljump{\term[3]}_{\tilde{\vsig}_{\com{\id}}}) = \hat{\val}(\ljump{\term[2] \compo \term[3]}_{\tilde{\vsig}_{\com{\id}}}).
    \end{align*}
\end{proof}
Particularly, for $\KAtermclass$ \kl{terms}, we have the following.
\begin{lem}[cf.\ \Cref{lemma: lang val}]\label{lemma: lang val KA}
    Let $\val \in \LANG$.
    For all $\KAtermclass$ \kl{terms} $\term$, we have: $\hat{\val}(\term) = \hat{\val}(\ljump{\term})$.
\end{lem}
\begin{proof}
    We have $\ljump{\term} = \ljump{\term}_{\tilde{\vsig}}$ since $\KAtermclass$ \kl{terms} do not contain complement.
    Hence, by \Cref{lemma: lang val}, this completes the proof.
\end{proof}

Additionally, by \Cref{lemma: lang val KA}, the converse direction of (\ref{equation: LANG to lang}) holds for $\KAtermclass$ \kl{terms} (cf.\ \Cref{remark: LANG and lang}).
The following is an explicit proof not relying on the completeness of KAs.
\begin{prop}\label{theorem: LANG and lang}
    For all $\KAtermclass$ \kl{terms} $\term[1], \term[2]$, we have:
    \[\LANG \models \term[1] = \term[2] \quad \Leftrightarrow \quad \ljump{\term[1]} = \ljump{\term[2]}.\]
\end{prop}
\begin{proof}
    We have:
    \begin{align*}
        \LANG \models \term[1] = \term[2] & \quad\Rightarrow\quad \ljump{\term[1]} = \ljump{\term[2]} \tag{$\val_{\mathrm{st}} \in \LANG$}                                 \\
                                            & \quad\Rightarrow\quad \forall \val \in \LANG, \hat{\val}(\ljump{\term[1]}) = \hat{\val}(\ljump{\term[2]})                       \\
                                            & \quad\Leftrightarrow\quad \forall \val \in \LANG, \hat{\val}(\term[1]) = \hat{\val}(\term[2]) \tag{\Cref{lemma: lang val KA}} \\
                                            & \quad\Leftrightarrow\quad  \LANG \models \term[1] = \term[2]. \tag{By definition}
    \end{align*}
    Hence, this completes the proof.
\end{proof}
\section{The identity inclusion problem}\label{section: identity}
We first consider the \emph{identity inclusion problem} w.r.t.\ languages:
\begin{center}
    Given a $\KAtermclass_{\set{\com{x}, \com{\id}}}$ \kl{term} $\term$, does $\LANG \models \id \le \term$?
\end{center}
This problem is relatively easily solvable.
Since $\LANG \models \id \le \term$ iff $\id \in \hat{\val}(\term)$ for all \kl{valuations} $\val \in \LANG$,
it suffices to consider the membership of the empty word $\eps$.
Thus, we have:
\begin{lem}\label{lemma: identity val}
    Let $\val, \val' \in \LANG$ be such that for all \kl{variables} $x$, $\eps \in \val(x)$ iff $\eps \in \val'(x)$.
    For all $\KAtermclass_{\set{\com{x}, \com{\id}}}$ \kl{terms} $\term$, we have: $\eps \in \hat{\val}(\term)$ iff $\eps \in \hat{\val}'(\term)$.
\end{lem}
\begin{proof}
    By \Cref{lemma: lang val}, it suffices to show when $\term$ is a \kl{word} over $\tilde{\vsig}_{\com{\id}}$.
    (If \Cref{lemma: identity val} is shown for \kl{words} over $\tilde{\vsig}_{\com{\id}}$, then by using \Cref{lemma: lang val}, for all $\KAtermclass_{\set{\com{x}, \com{\id}}}$ \kl{terms} $\term$, we have: $\eps \in \hat{\val}(\term)$ iff $(\exists \word \in \ljump{\term}_{\tilde{\vsig}}, \eps \in \hat{\val}(\word))$ iff $(\exists \word \in \ljump{\term}_{\tilde{\vsig}}, \eps \in \hat{\val}'(\word))$ iff $\eps \in \hat{\val}'(\term)$.)
    Let $\term = x_0 \dots x_{m-1}$ where $m \ge 0$ and $x_0, \dots, x_{m-1} \in \tilde{\vsig}_{\com{\id}}$.
    Then we have:
    \begin{align*}
        \eps \in \hat{\val}(\term) &\; \Leftrightarrow\;  (\forall k \in \range{0, m-1}, \eps \in \hat{\val}(x_k))\\
        &\; \Leftrightarrow\;  (\forall k \in \range{0, m-1}, \eps \in \hat{\val}'(x_k)) 
        \; \Leftrightarrow\;  \eps \in \hat{\val}'(\term).
    \end{align*}
    Hence, this completes the proof.
\end{proof}
By \Cref{lemma: identity val}, it suffices to consider a finite number of \kl{valuations}, as follows.
\begin{thm}\label{theorem: identity val}
    For all $\KAtermclass_{\set{\com{x}, \com{\id}}}$ \kl{terms} $\term$, we have:
    \[\LANG \models \id \le \term \quad\Leftrightarrow\quad \LANG_{0} \models \id \le \term.\] 
\end{thm}
\begin{proof}
    ($\Rightarrow$):
    By $\LANG_{0} \subseteq \LANG$.
    ($\Leftarrow$):
    We prove the contraposition.
    By $\LANG \not\models \id \le \term$,
    let $\val \in \LANG$ be s.t.\ $\hat{\val}(\id) \not\subseteq \hat{\val}(\term)$ (i.e., $\eps \not\in \hat{\val}(\term)$).
    Let $\val^{\tuple{}} \in \LANG_{0}$ be the \kl{valuation} defined by:
    \[\val^{\tuple{}}(x) \ \defeq\  \set{\eps \mid \eps \in \val(x)}.\]
    By \Cref{lemma: identity val}, we have $\eps \not\in \hat{\val}^{\tuple{}}(\term)$.
    Hence, $\hat{\val}^{\tuple{}}(\id) \not\subseteq \hat{\val}^{\tuple{}}(\term)$.
\end{proof}

Note that the \kl{equational theory} of $\LANG_{0}$ can be reduced to the \kl{equational theory} of Boolean algebra by the following fact.
\begin{prop}\label{prop: LANG0 and Boolean algebra}
    The $(\sig \setminus \set{\bl^{*}})$-reduct of the \kl{$\sig$-algebra} $\lang_{\emptyset}$ is isomorphic to the $2$-valued Boolean algebra,
    where $\id$ maps to the true constant, $\emp$ to the false constant, $\compo$ to the conjunction, $\union$ to the disjunction, and $\bl^{-}$ to the complement.
\end{prop}
\begin{proof}
    Easy, because the universe $\domain{\lang_{\emptyset}}$ is the two elements set $\set{\emptyset, \set{\eps}}$.
\end{proof}
Additionally, we can eliminate $\bl^{*}$ by using the \kl{equation} $\LANG_{0} \models \term^* = \id$.
We then have the following complexity result.
\begin{cor}\label{corollary: identity decidable}
    The identity inclusion problem---given a $\KAtermclass_{\set{\com{x}, \com{\id}}}$ \kl{term} $\term$, does $\LANG \models \id \le \term$?---is decidable and coNP-complete.
\end{cor}
\begin{proof}
    By \Cref{theorem: identity val,prop: LANG0 and Boolean algebra}, this problem is almost equivalent to the validity problem of propositional formulas in disjunctive normal form, which is a well-known coNP-complete problem \cite{cookComplexityTheoremprovingProcedures1971}.\footnote{From this, the \kl{equational theory} of $\LANG_{0}$ is decidable in coNP, even for $\KAtermclass_{\set{-}}$ \kl{terms}.}
    (in coNP):
    For the complement of this problem, \Cref{theorem: identity val} induces the following non-deterministic polynomial algorithm:
    \begin{enumerate}
        \item Pick up some $\val \in \LANG_{0}$ s.t.\ $\val(x) \subseteq \set{\eps}$ for each $x$, non-deterministically.
        \item If $\hat{\val}(\id) \not\subseteq \hat{\val}(\term)$, then return $\const{True}$;
              otherwise return $\const{False}$.
    \end{enumerate}
    Then we have $\begin{cases}
        \LANG \not\models \id \le \term & (\mbox{some execution returns $\const{True}$})\\
        \LANG \models \id \le \term & (\mbox{otherwise})
    \end{cases}$.
    Hence, the identity inclusion problem is decidable in coNP, as its complemented problem is in NP.

    (coNP-hard):
    Given a propositional formula $\varphi$ in disjunctive normal form,
    let $\term$ be the $\KAtermclass_{\set{\com{x}, \com{\id}}}$ \kl{term} obtained from $\varphi$ according to the map of \Cref{prop: LANG0 and Boolean algebra} (so, conjunction $\land$ maps to $\compo$, disjunction $\lor$ to $\union$, positive literal $x$ to the variable $x$, and negative literal $\com{x}$ to the complemented variable $\com{x}$);
    for example, if $\varphi = (x \land \com{y}) \lor (y \lor \com{x})$, then $\term = (x \compo \com{y}) \union (y \union \com{x})$.
    By \Cref{prop: LANG0 and Boolean algebra} and \Cref{theorem: identity val}, $\varphi$ is valid in propositional logic iff $\LANG_{0} \models \id \le \term$ iff $\LANG \models \id \le \term$.
    Hence, the identity inclusion problem is coNP-hard.
\end{proof}

\begin{rem}\label{remark: identity inclusion}
    Under the standard language valuation, the identity inclusion problem---given a $\KAtermclass_{\set{\com{x}, \com{\id}}}$ \kl{term} $\term$, does $\ljump{\id} \subseteq \ljump{\term}$ (i.e., $\eps \in \ljump{\term}$)?---is decidable in P, because we can compute ``$\eps \in \ljump{\term}$?'' by induction on $\term$, as $\eps \not\in \ljump{x}$ and $\eps \in \ljump{\com{x}}$ for every variable $x$.
    Hence, for $\KAtermclass_{\set{\com{x}, \com{\id}}}$ \kl{terms}, the identity inclusion problem w.r.t.\ \kl{languages} is strictly harder than that under the standard language valuation, unless P = NP.
    (This situation is the same for $\KAtermclass_{\set{\com{x}}}$ \kl{terms}.)
\end{rem}
\section{Words-to-letters valuations for the variable/word inclusion problem}\label{section: words_to_letters}
Next, we consider the \emph{variable inclusion problem}:
\begin{center}
    Given a \kl{variable} $x$ and a $\KAtermclass_{\set{\com{x}, \com{\id}}}$ \kl{term} $\term$, does $\LANG \models x \le \term$?
\end{center}
In the identity inclusion problem, if $\word \in \hat{\val}(\id) \setminus \hat{\val}(\term)$, then $\word = \eps$ should hold; so it suffices to consider the membership of the empty \kl{word} $\eps$.
However, in the variable inclusion problem, this situation changes; if $\word \in \hat{\val}(x) \setminus \hat{\val}(\term)$, then $\word$ is possibly any \kl{word}.
To overcome this problem, we introduce \kl{words-to-letters valuations} (\Cref{definition: val for word to letter,definition: val for words to letters}).

In \Cref{section: variable inclusion}, we consider the \emph{variable inclusion problem}.
In \Cref{section: word inclusion}, we consider the \emph{word inclusion problem}, which is a generalization of the \emph{variable inclusion problem} from \kl{variables} to \kl{words}.

\subsection{The variable inclusion problem}\label{section: variable inclusion}
Let $\word \in \hat{\val}(x) \setminus \hat{\val}(\term)$ be a non-empty \kl{word} $\word$.
Then we can construct a \kl{valuation} $\val'$ s.t. $\ell \in \hat{\val}'(x) \setminus \hat{\val}'(\term)$ for some \emph{\kl{letter}} $\ell$.
If such $\val'$ can be constructed from $\val$, then it suffices to consider the membership of \kl{letters}.
Such $\val'$ can be defined as follows:
\begin{defi}\label{definition: val for word to letter}
    For a \kl{valuation} $\val \in \LANG_{X}$ and a \kl{word} $\word$ over $X$,
    the \kl{valuation} $\val^{\word} \in \LANG_{\set{\ell}}$ (where $\ell$ is a \kl{letter}) is defined as follows:
    \[\val^{\word}(x) \ \defeq\  \set{\eps \mid \eps \in \val(x)} \cup \set{\ell \mid \word \in \val(x)}.\]
\end{defi}
In the following, when $\word$ is a non-empty \kl{word},
we prove that $\val^{\word}$ satisfies the condition of $\val'$ above, i.e., the following conditions:
\begin{itemize}
    \item $\word \in \hat{\val}(x) \quad\Rightarrow\quad \ell \in \hat{\val}^{\word}(x)$,
    \item $\word \not\in \hat{\val}(\term) \quad\Rightarrow\quad \ell \not\in \hat{\val}^{\word}(\term)$.
\end{itemize}
The first condition is clear by the definition of $\val^{\word}$.
The second condition is shown as follows.
\begin{lem}\label{lemma: ell abstraction}
    Let $\val \in \LANG$ and $\word$ be a non-empty \kl{word}.
    For all $\KAtermclass_{\set{\com{x}, \com{\id}}}$ \kl{terms} $\term$, we have:
    \[\ell \in \hat{\val}^{\word}(\term) \quad \Rightarrow \quad \word \in \hat{\val}(\term).\]
\end{lem}
\begin{proof}
    As with \Cref{lemma: identity val}, by \Cref{lemma: lang val}, it suffices to show when $\term$ is a \kl{word} over $\tilde{\vsig}_{\com{\id}}$.
    Let $\term = x_0 \dots x_{m-1}$ where $m \ge 0$ and $x_0, \dots, x_{m-1} \in \tilde{\vsig}_{\com{\id}}$.
    Then there is $i \in \range{0, m-1}$ s.t.\ 
    \begin{itemize}
        \item $\ell \in \hat{\val}^{\word}(x_i)$,
        \item $\eps \in \hat{\val}^{\word}(x_j)$ for $j \in \range{0, m-1} \setminus \set{i}$.
    \end{itemize}
    For $x_i$, we distinguish the following cases:
    \begin{itemize}
        \item Case $x_i = z, \com{z}$ where $z \in \vsig$:
        By the construction of $\val^{\word}$,
        we have that $\ell \in \hat{\val}^{\word}(z)$ iff $\word \in \hat{\val}(z)$.
        Similarly, we also have that $\ell \in \hat{\val}^{\word}(\com{z})$ iff $\word \in \hat{\val}(\com{z})$.
        \item Case $x_i = \com{\id}$:
        Because $\word$ is a non-empty \kl{word}, we have $\word \in \hat{\val}(\com{\id})$.
    \end{itemize}
    Hence, $\word \in \hat{\val}(x_i)$.
    For $x_j$, by \Cref{lemma: identity val} and $\eps \in \hat{\val}^{\word}(x_j)$, we have $\eps \in \hat{\val}(x_j)$.
    Thus, $\word \in \hat{\val}(\term)$.
\end{proof}
Thus $\val^{\word}$ satisfies the following:
\begin{cor}\label{corollary: ell abstraction variable}
    Let $\val \in \LANG$.
    For all \kl{variables} $x$ and $\KAtermclass_{\set{\com{x}, \com{\id}}}$ \kl{terms} $\term[1]$, we have:
    \begin{itemize}
        \item For a non-empty \kl{word} $\word$, if $\word \in \hat{\val}(x) \setminus \hat{\val}(\term[1])$, then $\ell \in \hat{\val}^{\word}(x) \setminus \hat{\val}^{\word}(\term[1])$.
        \item For a \kl{word} $\word$, if $\eps \in \hat{\val}(x) \setminus \hat{\val}(\term[1])$, then $\eps \in \hat{\val}^{\word}(x) \setminus \hat{\val}^{\word}(\term[1])$.%
    \end{itemize}
\end{cor}
\begin{proof}
    For the first statement:
    By the construction of $\val^{\word}$ and $\word \in \hat{\val}(x)$, we have $\ell \in \hat{\val}^{\word}(x)$.
    By \Cref{lemma: ell abstraction} and $\word \not\in \hat{\val}(\term[1])$, we have $\ell \not\in \hat{\val}^{\word}(\term[1])$.
    For the second statement:
    By \Cref{lemma: identity val}.
\end{proof}
\begin{thm}\label{theorem: char and univ variable}
    For all \kl{variables} $x$ and $\KAtermclass_{\set{\com{x}, \com{\id}}}$ \kl{terms} $\term[1]$, the following are equivalent:
    \begin{enumerate}
        \item \label{theorem: char and univ variable 1} $\LANG \models x \le \term[1]$,
        \item \label{theorem: char and univ variable 2} $\set{\val \in \LANG_{\set{\ell}} \mid \mbox{$\forall y \in \vsig, \val(y) \subseteq \set{\eps, \ell}$}} \models x \le \term[1]$,
        \item \label{theorem: char and univ variable 3} $\bigcup_{X} \set{\val^{\word} \mid \val \in \LANG_{X} \mbox{ and } \word \in X^{+}} \models x \le \term[1]$.
    \end{enumerate}
\end{thm}
\begin{proof}
    (\ref{theorem: char and univ variable 1})$\Rightarrow$(\ref{theorem: char and univ variable 2}):
    Trivial.
    (\ref{theorem: char and univ variable 2})$\Rightarrow$(\ref{theorem: char and univ variable 3}):
    Because $\hat{\val}^{\word}(y) \subseteq \set{\eps, \ell}$ for each $y$.
    (\ref{theorem: char and univ variable 3})$\Rightarrow$(\ref{theorem: char and univ variable 1}):
    The contraposition is shown by \Cref{corollary: ell abstraction variable}.
\end{proof}
\begin{cor}\label{corollary: variable decidable}
    The variable inclusion problem---given a variable $x$ and a $\KAtermclass_{\set{\com{x}, \com{\id}}}$ \kl{term} $\term$, does $\LANG \models x \le \term$?---is decidable and coNP-complete for $\KAtermclass_{\set{\com{x}, \com{\id}}}$ \kl{terms}.
\end{cor}
\begin{proof}
    (in coNP):
    By (\ref{theorem: char and univ variable 2}) of \Cref{theorem: char and univ variable},
    we can give an algorithm as with \Cref{corollary: identity decidable}.
    (coNP-hard):
    We give a reduction from the validity problem of propositional formulas in disjunctive normal form, as with \Cref{corollary: identity decidable}.
    Given a propositional formula $\varphi$ in disjunctive normal form,
    let $\term$ be the $\KAtermclass_{\set{\com{x}, \com{\id}}}$ \kl{term} such that $\varphi$ is valid iff $\LANG \models \id \le \term$,
    where $\term$ can be given by the translation in \Cref{corollary: identity decidable}.
    By using a fresh \kl{variable} $z$,
    we have the following:
    \[\LANG \models \id \le \term \quad\Leftrightarrow\quad \LANG \models z \le z \compo \term.\]
    For ($\Rightarrow$): By the congruence law.
    For ($\Leftarrow$): By the substitution law.
    Hence, the variable inclusion problem is coNP-hard.
\end{proof}

\begin{rem}\label{remark: composition-free}
    \Cref{corollary: ell abstraction variable} fails for general \kl{terms}.
    E.g., when $\val(x) = \set{a}$, we have:
    \begin{align*}
        a a \in \hat{\val}(x x), &  & \ell \not\in \hat{\val}^{a a}(x x).
    \end{align*}
    (Note that $\hat{\val}^{a a}(x x) = \emptyset$ holds, as $\hat{\val}^{a a}(x) = \emptyset$ by $\val(x) = \set{a}$.)
\end{rem}

\begin{rem}\label{remark: commutativity}
    \Cref{theorem: char and univ variable} fails for general \kl{equations},
    e.g., the \kl{inequation} $xy \le yx$ (see also \Cref{proposition: LANG1}).
\end{rem} 
\subsection{The word inclusion problem}\label{section: word inclusion}
We recall $\tilde{\vsig}_{\com{\id}} = \set{x, \com{x} \mid x \in \vsig} \cup \set{\com{\id}}$.
The \emph{word inclusion problem} is the following problem:
\begin{center}
    Given a \kl{word} $\word$ over $\tilde{\vsig}_{\com{\id}}$ and a $\KAtermclass_{\set{\com{x}, \com{\id}}}$ \kl{term} $\term$, does $\LANG \models \word \le \term$?
\end{center}
We can also solve this problem by generalizing the \kl{valuation} of \Cref{definition: val for word to letter}, as follows.
\begin{defi}[\intro*\kl{words-to-letters valuations}]\label{definition: val for words to letters}
    For a \kl{valuation} $\val \in \LANG_{X}$ and \kl{words} $\word_0, \dots, \word_{n-1}$ over $X$,
    the \kl{valuation} $\val^{\tuple{\word_0, \dots, \word_{n-1}}} \in \LANG_{\set{\ell_0, \dots, \ell_{n-1}}}$ is defined as follows
    where $n \ge 0$ and $\ell_{0}, \dots, \ell_{n-1}$ are pairwise distinct \kl{letters}:\footnote{The \kl{valuation} $\val^{\word}$ (\Cref{definition: val for word to letter}) is the case $n = 1$.
    The \kl{valuation} $\val^{\tuple{}}$ in \Cref{theorem: identity val} is the case $n = 0$.}
    \[\val^{\tuple{\word_0, \dots, \word_{n-1}}}(x) \ \defeq\  \set{\ell_i \dots \ell_{j-1} \mid 0 \le i \le j \le n \  \land \  \word_i \dots \word_{j-1} \in \val(x)}.\]
\end{defi}
Let $\mathrm{Subw}(\word)$ be the set of all subwords of $\word$.
Then note that $\val^{\tuple{\word_0, \dots, \word_{n-1}}}(x) \subseteq \mathrm{Subw}(\ell_{0} \dots \ell_{n-1})$.

By using \kl{words-to-letters valuations}, we can strengthen the decidability result in \Cref{section: variable inclusion} from \kl{variables} to \kl{words}.

\begin{lem}[cf.\ \Cref{lemma: ell abstraction}]\label{lemma: ell abstraction gen}
    Let $\val \in \LANG$ and $\word_0, \dots, \word_{n-1}$ be non-empty \kl{words} where $n \ge 0$.
    For all $\KAtermclass_{\set{\com{x}, \com{\id}}}$ \kl{terms} $\term$ and $0 \le i \le j \le n$, we have:
    \[\ell_{i} \dots \ell_{j-1} \in \hat{\val}^{\tuple{\word_0, \dots, \word_{n-1}}}(\term) \quad \Rightarrow \quad \word_{i} \dots \word_{j-1} \in \hat{\val}(\term).\]
\end{lem}
\begin{proof}
    By \Cref{lemma: lang val}, it suffices to show when $\term$ is a \kl{word} over $\tilde{\vsig}_{\com{\id}}$.
    Let $\term = x_0 \dots x_{m-1}$ where $m \ge 0$ and $x_0, \dots, x_{m-1} \in \tilde{\vsig}_{\com{\id}}$.
    Then there are $i = l_0 \le l_1 \le \dots \le l_{m-1} \le l_{m} = j$ s.t.\ 
    $\ell_{l_{k}} \dots \ell_{l_{k+1} - 1} \in \hat{\val}^{\tuple{\word_0, \dots, \word_{n-1}}}(x_k)$ for each $k \in \range{0, m-1}$.
    We distinguish the following cases:
    \begin{itemize}
        \item Case $x_k = z, \com{z}$ where $z \in \vsig$:
        By the construction of $\val^{\tuple{\word_0, \dots, \word_{n-1}}}$,
        we have that $\ell_{l_{k}} \dots \ell_{l_{k+1} - 1} \in \hat{\val}^{\tuple{\word_0, \dots, \word_{n-1}}}(z)$ iff $\word_{l_{k}} \dots \word_{l_{k+1} - 1} \in \hat{\val}(z)$.
        We also have that $\ell_{l_{k}} \dots \ell_{l_{k+1} - 1} \in \hat{\val}^{\tuple{\word_0, \dots, \word_{n-1}}}(\com{z})$ iff $\word_{l_{k}} \dots \word_{l_{k+1} - 1} \in \hat{\val}(\com{z})$.
        \item Case $x_k = \com{\id}$:
        By $\eps \not\in \hat{\val}^{\tuple{\word_0, \dots, \word_{n-1}}}(\com{\id})$, we have $l_{k} < l_{k+1}$, and thus $\word_{l_{k}} \dots \word_{l_{k+1} - 1}$ is a non-empty \kl{word}.
        Thus, we have $\word_{l_{k}} \dots \word_{l_{k+1} - 1} \in \hat{\val}(\com{\id})$.
    \end{itemize}
    Thus, we have $\word_{l_{k}} \dots \word_{l_{k+1} - 1} \in \hat{\val}(x_k)$.
    Hence, we have $\word_{i} \dots \word_{j - 1} \in \hat{\val}(\term)$.
\end{proof}
Moreover, we have the following.
\begin{lem}\label{lemma: ell abstraction word}
    Let $\val \in \LANG$.
    Let $\word[2] = x_0 \dots x_{n-1}$ be a \kl{word} over $\tilde{\vsig}_{\com{\id}}$ and let $\word \in \hat{\val}(\word[2])$.
    Then there are $0 \le m \le n$ and non-empty \kl{words} $\word_{0}, \dots, \word_{m-1}$ such that
    $\word = \word_{0} \dots \word_{m-1}$ and $\ell_0 \dots \ell_{m-1} \in \hat{\val}^{\tuple{\word_0, \dots, \word_{m-1}}}(\word[2])$.
\end{lem}
\begin{proof}
    By $\word \in \hat{\val}(\word[2])$,
    let $\word = \word'_0 \dots \word'_{n-1}$ be s.t.\ $\word'_k \in \hat{\val}(x_k)$ for each $k$.
    Let $\tuple{\word_0, \dots, \word_{m-1}}$ be the sequence $\tuple{\word'_0, \dots, \word'_{n-1}}$ in which empty \kl{words} are eliminated.
    Let $f$ be the corresponding map such that $\word_{k} = \word'_{f(k)}$.
    By the construction of $\val^{\tuple{\word_0, \dots, \word_{m-1}}}$ and $\word'_{f(k)} \in \hat{\val}(x_{f(k)})$, we have $\ell_k \in \hat{\val}^{\tuple{\word_0, \dots, \word_{m-1}}}(x_{f(k)})$.
    Also, $\eps \in \hat{\val}(x_{k})$ implies $\eps \in \hat{\val}^{\tuple{\word_0, \dots, \word_{m-1}}}(x_{k})$.
    Thus, we have $\ell_{0} \dots \ell_{m-1} \in \hat{\val}^{\tuple{\word_0, \dots, \word_{m-1}}}(\word[2])$.
\end{proof}

\begin{thm}[cf.\ \Cref{theorem: char and univ variable}]\label{theorem: char and univ word}
    Let $\word[2] = x_0 \dots x_{n-1}$ be a \kl{word} over $\tilde{\vsig}_{\com{\id}}$ and let $\term$ be a $\KAtermclass_{\set{\com{x}, \com{\id}}}$ \kl{term}.
    The following are equivalent:
    \begin{enumerate}
        \item \label{theorem: char and univ word 1} $\LANG \models \word[2] \le \term$,

        \item \label{theorem: char and univ word 2} $\bigcup_{m \le n} \set{\val \in \LANG_{\set{\ell_{0}, \dots, \ell_{m-1}}} \mid \forall x,  \val(x) \subseteq \mathrm{Subw}(\ell_{0} \dots \ell_{m-1})} \models \word[2] \le \term$, %
        \item \label{theorem: char and univ word 3} $\bigcup_{X} \bigcup_{m \le n} \set{\val^{\tuple{\word_0, \dots, \word_{m-1}}} \mid \val \in \LANG_{X} \mbox{ and } \word_0, \dots, \word_{m-1} \in X^{+}} \models \word[2] \le \term$.
    \end{enumerate}
\end{thm}
\begin{proof}
    (\ref{theorem: char and univ word 1})$\Rightarrow$(\ref{theorem: char and univ word 2}):
    Trivial.
    (\ref{theorem: char and univ word 2})$\Rightarrow$(\ref{theorem: char and univ word 3}):
    Because $\hat{\val}^{\tuple{\word_0, \dots, \word_{m-1}}}(x) \subseteq \set{\ell_{i} \dots \ell_{j-1} \mid 0 \le i \le j \le m}$ holds for each $x$.
    (\ref{theorem: char and univ word 3})$\Rightarrow$(\ref{theorem: char and univ word 1}):
    We show the contraposition.
    Let $\word \in \hat{\val}(\word[2]) \setminus \hat{\val}(\term[1])$.
    By \Cref{lemma: ell abstraction word}, there are $0 \le m \le n$ and non-empty \kl{words} $\word_{0}, \dots, \word_{m-1}$ such that
    $\word = \word_{0} \dots \word_{m-1}$ and $\ell_0 \dots \ell_{m-1} \in \hat{\val}^{\tuple{\word_0, \dots, \word_{m-1}}}(\word[2])$.
    By $\word \not\in \hat{\val}(\term[1])$ and \Cref{lemma: ell abstraction gen}, we have $\ell_0 \dots \ell_{m-1} \not\in \hat{\val}^{\tuple{\word_0, \dots, \word_{m-1}}}(\term[1])$.
    Hence, this completes the proof.
\end{proof}

\begin{cor}[cf.\ \Cref{corollary: variable decidable}]\label{corollary: word decidable}
    The word inclusion problem---given a \kl{word} $\word$ and a \kl{term} $\term$, does $\LANG \models \word \le \term$?---is decidable and coNP-complete for $\KAtermclass_{\set{\com{x}, \com{\id}}}$ \kl{terms}.
\end{cor}
\begin{proof}
    (coNP-hard):
    By \Cref{corollary: identity decidable}, as $\word$ is possibly $\const{I}$.
    (in coNP):
    By (\ref{theorem: char and univ word 2}) of \Cref{theorem: char and univ word},
    we can give an algorithm as with \Cref{corollary: variable decidable}.
\end{proof} 
\subsection{Generalization for terms of bounded length}\label{section: bounded length}
We can generalize the argument in \Cref{section: variable inclusion,section: word inclusion} for more general problems.
For a $\KAtermclass_{\set{\com{x}, \com{\id}}}$ \kl{term} $\term$, we define the \kl{supremum length} $\langlen(\term) \in \nat \cup \set{\omega}$ as follows:
\[\langlen(\term) \ \defeq\  \sup(\set{\|\word\| \mid \word \in \ljump{\term}_{\tilde{\vsig}_{\com{\id}}}} \cup \set{0})\]
where $\omega$ denotes the smallest infinite ordinal.

\begin{lem}\label{lemma: bounded length}
    Let $\term$ be a $\KAtermclass_{\set{\com{x}, \com{\id}}}$ \kl{term}.
    Let $\val \in \LANG$ and let $\word \in \hat{\val}(\term)$.
    Then there are $0 \le m \le \langlen(\term[1])$ and non-empty \kl{words} $\word_{0}, \dots, \word_{m-1}$ s.t.\ $\word = \word_{0} \dots \word_{m-1}$ and $\ell_{0} \dots \ell_{m-1} \in \hat{\val}^{\tuple{\word_{0}, \dots, \word_{m-1}}}(\term)$.
\end{lem}
\begin{proof}
    By \Cref{lemma: lang val},
    there is a \kl{word} $\word[2] \in \ljump{\term}_{\tilde{\vsig}_{\com{\id}}}$ such that $\word \in \hat{\val}(\word[2])$.
    By $\|\word[2]\| \le \langlen(\term[1])$ and \Cref{lemma: ell abstraction word}, this completes the proof.
\end{proof}
Thus, we have the following.
\begin{thm}[cf.\ \Cref{theorem: char and univ word}]\label{theorem: char and univ gen}
    Let $\term[1]$ and $\term[2]$ be $\KAtermclass_{\set{\com{x}, \com{\id}}}$ \kl{terms}.
    The following are equivalent:
    \begin{enumerate}
        \item \label{theorem: char and univ gen 1} $\LANG \models \term[1] \le \term[2]$,
        \item \label{theorem: char and univ gen 2} $\bigcup_{m \le \langlen(\term[1])} \set{\val \in \LANG_{\set{\ell_{0}, \dots, \ell_{m-1}}} \mid \forall x,  \val(x) \subseteq \mathrm{Subw}(\ell_{0} \dots \ell_{m-1})} \models \term[1] \le \term[2]$, %
        \item \label{theorem: char and univ gen 3} $\bigcup_{X} \bigcup_{m \le \langlen(\term[1])} \set{\val^{\tuple{\word_0, \dots, \word_{m-1}}} \mid \val \in \LANG_{X} \mbox{ and } \word_0, \dots, \word_{m-1} \in X^{+}} \models \term[1] \le \term[2]$.
    \end{enumerate}
\end{thm}
\begin{proof}
    As with \Cref{theorem: char and univ word}, by using \Cref{lemma: bounded length} instead of \Cref{lemma: ell abstraction word}.
\end{proof}

We say that a \kl{term} $\term$ is \intro*\kl{star-free} if the Kleene-star ($\bl^{\kstar}$) does not occur in $\term$.
By \Cref{theorem: char and univ gen}, we have the following.
\begin{cor}\label{corollary: star-free}
    The following problem is coNP-complete:
    \begin{center}
        Given a \kl{star-free} $\KAtermclass_{\set{\com{x}, \com{\id}}}$ \kl{term} $\term[1]$ and a $\KAtermclass_{\set{\com{x}, \com{\id}}}$ \kl{term} $\term[2]$, does $\LANG \models \term[1] \le \term[2]$?
    \end{center}
\end{cor}
\begin{proof}
    (coNP-hard):
    By \Cref{corollary: identity decidable}, as $\term[1]$ is possibly $\id$.
    (in coNP):
    Because $\term[1]$ is \kl{star-free}, we have $\langlen(\term[1]) \le \|\term[1]\|$.
    By (\ref{theorem: char and univ gen 2}) of \Cref{theorem: char and univ gen}, we can give an algorithm as with \Cref{corollary: word decidable}.
\end{proof}
Moreover, we have the following as a corollary.
\begin{cor}[bounded alphabet property]\label{corollary: bounded alphabet}
    Let $\term[1]$ and $\term[2]$ be $\KAtermclass_{\set{\com{x}, \com{\id}}}$ \kl{terms}.
    Then we have:
    \[\LANG \models \term[1] \le \term[2] \quad\Leftrightarrow\quad \LANG_{\langlen(\term[1])} \models \term[1] \le \term[2].\]
\end{cor}
\begin{proof}
    By \Cref{theorem: char and univ gen}.
\end{proof}

\subsection{The universality problem}
The \emph{universality problem} w.r.t.\ $\LANG$ is the following problem:
\begin{center}
    Given a $\KAtermclass_{\set{\com{x}, \com{\id}}}$ \kl{term} $\term$, does $\LANG \models \top \le \term$?
\end{center}
Interestingly, the universality problem of $\LANG$ is decidable and coNP-complete.
\begin{cor}\label{corollary: universality decidable}
    The universality problem is coNP-complete for $\KAtermclass_{\set{\com{x}, \com{\id}}}$ \kl{terms}.
\end{cor}
\begin{proof}
    (in coNP):
    We have that $\LANG \models \top = x \union \com{x}$ and $\langlen(x \union \com{x}) = 1$.
    Thus, by (\ref{theorem: char and univ gen 2}) of \Cref{theorem: char and univ gen}, we can give an algorithm as with \Cref{corollary: word decidable}.
    (coNP-hard):
    We give a reduction from the validity problem of propositional formulas in disjunctive normal form, as with \Cref{corollary: identity decidable,corollary: variable decidable}.
    Given a propositional formula $\varphi$ in disjunctive normal form,
    let $\term$ be the $\KAtermclass_{\set{\com{x}, \com{\id}}}$ \kl{term} such that $\varphi$ is valid iff $\LANG \models \id \le \term$
    where $\term$ is obtained by the translation in \Cref{corollary: identity decidable}.
    Then we have:
    \[\LANG \models \id \le \term \quad\Leftrightarrow\quad \LANG \models \top \le \top \compo \term.\]
    For ($\Rightarrow$):
    By the congruence law.
    For ($\Leftarrow$):
    By $\LANG \models \id \le \top \compo \term$
    and that $\LANG \models \id \le \term[2] \compo \term[3]$ iff $\LANG \models \id \le \term[2]$ and $\LANG \models \id \le \term[3]$ for any $\term[2], \term[3]$.
    Hence, the universality problem is coNP-hard.
\end{proof}

\begin{rem}
    In the standard language equivalence, the universality problem is usually of the form $\ljump{\vsig^{\kstar}} = \ljump{\term}$,
    as $\ljump{\vsig^{\kstar}} = \ljump{\top}$ (when $\vsig$ is finite) and the constant $\top$ is usually not a primitive symbol of regular expressions.
    However, $\LANG \models \vsig^{\kstar} \le \term$ is different from $\LANG \models \top \le \term$, as $\LANG \not \models \vsig^{\kstar} = \top$.
\end{rem}

\begin{rem}\label{remark: universality}
    Under the standard language equivalence, the universality problem---given a term $\term$, does $\ljump{\top} \subseteq \ljump{\term}$? (i.e., $\ljump{\term} = \vsig^*$?)---is PSPACE-hard \cite{meyerEquivalenceProblemRegular1972, Meyer1973, Hunt1976}.
    Hence, for $\KAtermclass_{\set{\com{x}, \com{\id}}}$ \kl{terms}, the universality problem w.r.t.\ languages is strictly easier (cf.\ \Cref{remark: identity inclusion}) than that under the standard language equivalence unless NP = PSPACE.
\end{rem}

\subsection{Words-to-letters valuation property}
As an immediate consequence of \Cref{theorem: char and univ gen},
we have that \kl{words-to-letters valuations} are sufficient for the \kl{equational theory w.r.t.\ languages} for $\KAtermclass_{\set{\com{x}, \com{\id}}}$ \kl{terms}.
\begin{cor}[words-to-letters valuation property]\label{corollary: word witness}
    For all $\KAtermclass_{\set{\com{x}, \com{\id}}}$ \kl{terms} $\term[1], \term[2]$, the following are equivalent:
    \begin{enumerate}
        \item \label{corollary: word witness 1} $\LANG \models \term[1] \le \term[2]$,
        \item \label{corollary: word witness 3} $\bigcup_{X}\bigcup_{m \in \nat} \set{\val^{\tuple{\word_0, \dots, \word_{m-1}}} \mid \val \in \LANG_{X} \mbox{ and } \word_0, \dots, \word_{m-1} \in X^{+}} \models \term[1] \le \term[2]$.
    \end{enumerate}
\end{cor}
\begin{proof}
    By \Cref{theorem: char and univ gen}, as $\langlen(\term[1]) \le \omega$.
\end{proof}

Additionally, \Cref{corollary: word witness} also shows the following property.
\begin{cor}\label{corollary: countably infinite alphabet}
    For all $\KAtermclass_{\set{\com{x}, \com{\id}}}$ \kl{terms} $\term[1], \term[2]$, we have:
    \[\LANG \models \term[1] \le \term[2] \quad\Leftrightarrow\quad \LANG_{\aleph_0} \models \term[1] \le \term[2].\]
\end{cor}
\begin{proof}
    By \Cref{corollary: word witness}.
\end{proof}
We can show this property, moreover, for $\KAtermclass_{\set{-}}$ \kl{terms}, by using the following transformation of \kl{valuations}.
\begin{lem}\label{lemma: bounding alphabet}
    Let $\val \in \LANG_{A}$.
    Let $B \subseteq A$.
    Let $\val_{B} \in \LANG_{B}$ be the \kl{valuation} defined by $\val_{B}(x) = \val(x) \cap B^*$ for each $x \in \vsig$.
    For all $\KAtermclass_{\set{-}}$ \kl{terms} $\term$, we have $\hat{\val}_{B}(\term) = \hat{\val}(\term) \cap B^*$.
\end{lem}
\begin{proof}
    By easy induction on $\term$, using the following equivalences:
    \begin{align*}
        \label{lemma: bounding alphabet cup} (\la[1] \cap B^*) \cup (\la[2] \cap B^*) \quad&=\quad (\la[1] \cup \la[2]) \cap B^*, \tag{\Cref{lemma: bounding alphabet}-($\cup$)}\\ 
        \label{lemma: bounding alphabet cdot} (\la[1] \cap B^*) \compo (\la[2] \cap B^*) \quad&=\quad (\la[1] \compo \la[2]) \cap B^*, \tag{\Cref{lemma: bounding alphabet}-($\compo$)}\\
        \label{lemma: bounding alphabet compl} B^* \setminus (\la[1] \cap B^*) \quad&=\quad (B^* \setminus \la[1]) \cap B^*. \tag{\Cref{lemma: bounding alphabet}-($\bl^{-}$)}
    \end{align*}
    
    Case $\term = x, \com{x}$:
    By definition of $\val_{B}$.

    Case $\term = \emp, \id, \com{\id}$:
    By $\hat{\val}_{B}(\emp) = \emptyset$, $\hat{\val}_{B}(\id) = \set{\eps}$, and $\hat{\val}_{B}(\com{\id}) = B^* \setminus \set{\eps}$.

    Case $\term = \term[2] \union \term[3]$:
    We have:
    \begin{align*}
        \hat{\val}_{B}(\term[2] \union \term[3]) = \hat{\val}_{B}(\term[2]) \cup \hat{\val}_{B}(\term[3]) & = (\hat{\val}(\term[2]) \cap B^*) \cup (\hat{\val}(\term[3]) \cap B^*) \tag{IH}                                                                                                   \\
                                                                                                                 & = \hat{\val}(\term[2] \union \term[3]) \cap B^*. \tag{\ref{lemma: bounding alphabet cup}}
    \end{align*}

    Case $\term = \term[2] \compo \term[3]$:
    We have:
    \begin{align*}
        \hat{\val}_{B}(\term[2] \compo \term[3]) = \hat{\val}_{B}(\term[2]) \compo \hat{\val}_{B}(\term[3]) &= (\hat{\val}(\term[2]) \cap B^*) \compo (\hat{\val}(\term[3]) \cap B^*) \tag{IH}\\
                                                                                                    & = (\hat{\val}(\term[2] \compo \term[3])) \cap B^*. \tag{\ref{lemma: bounding alphabet cdot}}
    \end{align*}

    Case $\term = \term[2]^*$:
    We have:
    \begin{align*}
        \hat{\val}_{B}(\term[2]^*) = \bigcup_{n \in \nat} \hat{\val}_{B}(\term[2])^n &= \bigcup_{n \in \nat} (\hat{\val}(\term[2]) \cap B^*)^n \tag{IH}\\
                                                                               &= (\bigcup_{n \in \nat} \hat{\val}(\term[2])^n) \cap B^* \tag{\ref{lemma: bounding alphabet cup}, \ref{lemma: bounding alphabet cdot}}\\
                                                                               &= \hat{\val}(\term[2]^*) \cap B^*. 
    \end{align*}

    Case $\term = \term[2]^-$:
    We have:
    \begin{align*}
        \hat{\val}_{B}(\term[2]^-) = B^* \setminus \hat{\val}_{B}(\term[2]) &= B^* \setminus (\hat{\val}(\term[2]) \cap B^*) \tag{IH}\\
                                                                      &= (B^* \setminus \hat{\val}(\term[2])) \cap B^* \tag{\ref{lemma: bounding alphabet compl}}\\
                                                                      &= \hat{\val}(\term[2]^-) \cap B^*.
    \end{align*}
    Hence, this completes the proof.
\end{proof}
\begin{cor}[countably infinite alphabet property]\label{corollary: countably infinite alphabet gen}
    For all $\KAtermclass_{\set{-}}$ \kl{terms} $\term[1], \term[2]$, we have:
    \[\LANG \models \term[1] \le \term[2] \quad\Leftrightarrow\quad \LANG_{\aleph_0} \models \term[1] \le \term[2].\]
\end{cor}
\begin{proof}
    ($\Rightarrow$):
    By $\LANG_{\aleph_0} \subseteq \LANG$.
    ($\Leftarrow$):
    We show the contraposition.
    Let $\val \in \LANG$ and let $a_0 \dots a_{n-1} \in \hat{\val}(\term[1]) \setminus \hat{\val}(\term[2])$.
    By \Cref{lemma: bounding alphabet}, we have $a_0 \dots a_{n-1} \in \hat{\val}_{B}(\term[1]) \setminus \hat{\val}_{B}(\term[2])$ where $B = \set{a_0, \dots, a_{n-1}}$.
    By $\val_{B} \in \LANG_{\aleph_0}$, this completes the proof.
\end{proof}

\begin{rem}
    To prove \Cref{corollary: countably infinite alphabet gen},
    it suffices to use ``\intro*\kl{letters-to-letters valuations}'', which are \kl{words-to-letters valuations} $\val^{\tuple{\word_0, \dots, \word_{m-1}}}$ where $\word_{0}, \dots, \word_{m-1}$ are restricted to \kl{letters}.
    Nevertheless, the transformation in \Cref{lemma: bounding alphabet} has better bounds of the number of \kl{letters}.
    For example, when $\word = \mathtt{a} \mathtt{b} \mathtt{a} \mathtt{b} \mathtt{a} \in \hat{\val}(\term[1]) \setminus \hat{\val}(\term[2])$, we have $\val_{\set{\mathtt{a}, \mathtt{b}}} \in \LANG_{2}$ (because the number of \kl{letters} occurring in $\word$ is $2$)
    and we have $\val^{\tuple{\mathtt{a}, \mathtt{b}, \mathtt{a}, \mathtt{b}, \mathtt{a}}} \in \LANG_{5}$ (because the \kl{length} of $\word$ is $5$).
\end{rem}

\section{On the hierarchy of $\LANG_n$}\label{section: hierarchy}
In this section, we consider \kl{equational theories} of $\LANG_{n}$ where $n$ is bounded.
First, even for KA \kl{terms}, the \kl{equational theories} of $\LANG_{0}$ and $\LANG_{1}$ are different.
Recall that the \kl{equational theory} of $\LANG_{0}$ corresponds to \kl[equational theory]{that} of Boolean algebra (\Cref{prop: LANG0 and Boolean algebra}).
\begin{prop}\label{proposition: LANG0}
    We have:
    \begin{itemize}
        \item $\LANG_{0} \models x \le \id$,
        \item $\LANG_{1} \not\models x \le \id$.
    \end{itemize}
\end{prop}
\begin{proof}
    For $\LANG_{0} \models x \le \id$: 
    Because $\hat{\val}(x) \subseteq \set{\eps} = \hat{\val}(\id)$ for all $\val \in \LANG_{0}$.
    For $\LANG_{1} \not\models x \le \id$:
    We have $\ell \in \hat{\val}(x) \setminus \hat{\val}(\id)$ when $\val(x) = \set{\ell}$.
\end{proof}
The \kl{equation} $x \com{x} \le \emp$ is another example: $\LANG_{0} \models x \com{x} \le \emp$ and $\LANG_{1} \not\models x \com{x} \le \emp$.

The \kl{equational theories} of $\LANG_{1}$ and $\LANG_{2}$ are also different, as follows.
\begin{prop}\label{proposition: LANG1}
    When $x, y \in \vsig$ are distinct, we have:
    \begin{itemize}
        \item $\LANG_{1} \models xy \le yx$,
        \item $\LANG_{2} \not\models xy \le yx$.
    \end{itemize}
\end{prop}
\begin{proof}
    For $\LANG_{1} \models xy \le yx$: 
    We have $\hat{\val}(x y) = \hat{\val}(y x)$, by the commutative law.
    For $\LANG_{2} \not\models xy \le yx$: 
    When $\val(x) = \set{\mathtt{a}}$ and $\val(y) = \set{\mathtt{b}}$,
    we have $\mathtt{a}\mathtt{b} \in \hat{\val}(xy) \setminus \hat{\val}(yx)$.
\end{proof}
Additionally, we recall that the \kl{equational theories} of $\LANG_{\aleph_0}$ and $\LANG$ are the same (\Cref{corollary: countably infinite alphabet gen}), even for $\KAtermclass_{\set{-}}$ \kl{terms}.

Now, what about the \kl{equational theories} of $\LANG_{n}$ and $\LANG_{n+1}$ for $n \ge 2$?
In this section, we show that this depends on the class of \kl{terms}, as follows.
\begin{itemize}
    \item For $\KAtermclass$ \kl{terms},
    the \kl{equational theory} of $\LANG_{n}$ coincides with \kl[equational theory]{that} of $\LANG_{n+1}$ (\Cref{theorem: KA collapse} in \Cref{section: KA collapse}),
    
    \item For $\KAtermclass_{\set{-}}$ (i.e., $\KAtermclass$ with full complement) \kl{terms},
    the \kl{equational theory} of $\LANG_{n}$ does not coincide with \kl[equational theory]{that} of $\LANG_{n+1}$ (\Cref{theorem: separation gen} in \Cref{section: KA- infinite}).
\end{itemize}
(We leave open for $\KAtermclass_{\set{\com{x}, \com{\id}}}$ \kl{terms}, see also \Cref{remark: open strictness}.)

\subsection{The hierarchy is collapsed for $\KAtermclass$ terms}\label{section: KA collapse}
For $\KAtermclass$ \kl{terms}, it is easy to see that the hierarchy of $\LANG_{n}$ is collapsed, as standard binary encodings work for $\KAtermclass$ \kl{terms}.
\begin{prop}\label{theorem: KA collapse}
    Let $n \in \nat$ where $n \ge 2$.
    For all $\KAtermclass$ \kl{terms} $\term[1]$, $\term[2]$, we have:
    \[\LANG_{n} \models \term[1] \le \term[2] \quad\Leftrightarrow\quad \LANG_{2} \models \term[1] \le \term[2].\]
\end{prop}
\begin{proof}[Proof Sketch]
    ($\Rightarrow$):
    By $\LANG_{2} \subseteq \LANG_{n}$.
    ($\Leftarrow$):
    Let $A = \set{\ell_0, \dots, \ell_{n-1}}$ and $B = \set{\const{a}, \const{b}}$.
    Let $f \colon A^* \to B^*$ be the unique monoid homomorphism extending $\ell_{i} \mapsto \mathtt{a} \mathtt{b}^{i}$ and let $f' \colon \wp(A^*) \to \wp(B^*)$ be the map: $f'(L) \defeq \set{f(\word) \mid \word \in A^*}$.
    Then, as $f'$ is an injective $\KAtermclass$-homomorphism (i.e., $f'$ preserves the operators $\union$, $\compo$, $\bl^{\kstar}$, $\emp$, and $\id$) from $\lang_{A}$ to $\lang_{B}$,
    we can show this proposition.
\end{proof}
Thus, for $\KAtermclass$ \kl{terms}, we have:
\begin{align*}
    & \mathrm{EqT}(\LANG_{0}) \supsetneq \mathrm{EqT}(\LANG_{1})\\
    & \supsetneq \mathrm{EqT}(\LANG_{2}) = \dots = \mathrm{EqT}(\LANG_{n}) = \dots = \mathrm{EqT}(\LANG_{\aleph_0}) = \mathrm{EqT}(\LANG).
\end{align*}
Here, $\mathrm{EqT}(\algclass)$ denotes the \kl{equational theory} of a class $\algclass$ for $\KAtermclass$ \kl{terms}.

\begin{rem}\label{remark: failure of KA collapse}
    We cannot directly extend \Cref{theorem: KA collapse} for $\KAtermclass_{\set{\com{x}}}$, $\KAtermclass_{\set{\com{\id}}}$, and $\KAtermclass$ with top \kl{terms},
    as the map $f'$ does not preserve the operators $\bl^{-}$ or $\top$.
\end{rem}

\subsection{The hierarchy is infinite for $\KAtermclass_{\set{-}}$ terms}\label{section: KA- infinite}
We first show that the \kl{equational theories} of $\LANG_{2}$ and $\LANG_{3}$ are not the same for $\KAtermclass_{\set{-}}$ \kl{terms},
and then we generalize the construction for the \kl{equational theories} of $\LANG_{n}$ and $\LANG_{n+1}$.
\begin{prop}\label{proposition: separation}
    Let $\term[1]$ and $\term[2]$ be the following $\KAtermclass_{\set{-}}$ \kl{terms}:
    \begin{align*}
        \term[1] &\ \defeq\  (\top ((x \union y \union z)^{*})^{-} \top)^{-},\\
        \term[2] &\ \defeq\  (\top ((x \union y)^{*})^{-} \top)^{-} \union (\top ((y \union z)^{*})^{-} \top)^{-} \union (\top ((z \union x)^{*})^{-} \top)^{-}.
    \end{align*}
    Then we have:
    \begin{itemize}
        \item $\LANG_{2} \models \term[1] \le \term[2]$,
        \item $\LANG_{3} \not\models \term[1] \le \term[2]$.
    \end{itemize}
\end{prop}
\begin{proof}
    For $\LANG_{2} \models \term[1] \le \term[2]$:
    Let $\val \in \LANG_{A}$ where $\#A = 2$.
    Let $\word\in \hat{\val}(\term[1])$.
    Let $B = \set{a \in A \mid \mbox{$a$ occurs in $\word$}}$.
    For each $a \in B$, if $a \not\in \hat{\val}(x \union y \union z)$, then by $a \not\in \hat{\val}((x \union y \union z)^*)$, we have $\word \in \hat{\val}(\top ((x \union y \union z)^{*})^{-} \top)$, and thus $\word \not\in \hat{\val}(\term[1])$, reaching a contradiction.
    Hence, $B \subseteq \hat{\val}(x \union y \union z)$.
    Because $\# B \le 2$, we have either one of the following:
    \begin{gather*}
        B \subseteq \hat{\val}(x \union y), \qquad 
        B \subseteq \hat{\val}(y \union z), \qquad 
        B \subseteq \hat{\val}(z \union x).
    \end{gather*}
    If $B \subseteq \hat{\val}(x \union y)$, then by $B^* \subseteq \hat{\val}((x \union y)^*)$, any \kl{word} in $\hat{\val}(((x \union y)^*)^{-})$ should contain some \kl{letter} in $A \setminus B$.
    Thus by $\word \not\in \hat{\val}(\top ((x \union y)^*)^{-} \top)$, we have $\word \in \hat{\val}(\term[2])$.
    Similarly for the other cases, we have $\word \in \hat{\val}(\term[2])$.
    Hence, this completes the proof.

    For $\LANG_{3} \not\models \term[1] \le \term[2]$:
    Let $A = \set{\mathtt{a}, \mathtt{b}, \mathtt{c}}$ and let $\val \in \LANG_{A}$ be the \kl{valuation} s.t.\
    $\val(x) = \set{\mathtt{a}}$, $\val(y) = \set{\mathtt{b}}$, and $\val(z) = \set{\mathtt{c}}$.
    Then we have:
    \begin{align*}
        \hat{\val}(\term[1]) &= \set{\mathtt{a}, \mathtt{b}, \mathtt{c}}^*, &
        \hat{\val}(\term[2]) &= \set{\mathtt{a}, \mathtt{b}}^* \cup \set{\mathtt{b}, \mathtt{c}}^* \cup \set{\mathtt{c}, \mathtt{a}}^*.
    \end{align*}
    Hence by $\val \not\models \term[1] \le \term[2]$, this completes the proof.
\end{proof}
We can straightforwardly generalize the argument above for separating the \kl{equational theory} of $\LANG_{n}$ and \kl[equational theory]{that} of $\LANG_{n+1}$, as follows:
\begin{thm}\label{theorem: separation gen}
    Let $n \ge 1$.
    Let $\term[1]$ and $\term[2]$ be the following $\KAtermclass_{\set{-}}$ \kl{terms}:
    \begin{align*}
        \term[1] &\ \defeq\  (\top ((\sum_{i \in \range{0, n}} x_i)^{*})^{-} \top)^{-}, & 
        \term[2] &\ \defeq\  \sum_{j \in \range{0, n}} (\top ((\sum_{i \in \range{0, n} \setminus \set{j}} x_i)^{*})^{-} \top)^{-}.
    \end{align*}
    Then we have:
    \begin{itemize}
        \item $\LANG_{n} \models \term[1] \le \term[2]$,
        \item $\LANG_{n+1} \not\models \term[1] \le \term[2]$.
    \end{itemize}
\end{thm}
\begin{proof}
    For $\LANG_{n} \models \term[1] \le \term[2]$:
    Let $\val \in \LANG_{A}$ where $\#A = n$.
    Let $\word\in \hat{\val}(\term[1])$.
    Let $B = \set{a \in A \mid \mbox{$a$ occurs in $\word$}}$.
    For each $a \in B$, if $a \not\in \hat{\val}(\sum_{i \in \range{0, n}} x_i)$, then by $a \not\in \hat{\val}((\sum_{i \in \range{0, n}} x_i)^*)$, we have $\word \in \hat{\val}(\top ((\sum_{i \in \range{0, n}} x_i)^{*})^{-} \top)$, and thus $\word \not\in \hat{\val}(\term[1])$, reaching a contradiction.
    Hence, $B \subseteq \hat{\val}(\sum_{i \in \range{0, n}} x_i)$.
    Because $\# B \le n$, there is some $j \in \range{0, n}$ s.t.\ 
    \begin{gather*}
        B \subseteq \hat{\val}(\sum_{i \in \range{0, n} \setminus \set{j}} x_i).
    \end{gather*}
    Then by $B^* \subseteq \hat{\val}((\sum_{i \in \range{0, n} \setminus \set{j}} x_i)^*)$, any \kl{word} in $\hat{\val}(((\sum_{i \in \range{0, n} \setminus \set{j}} x_i)^*)^{-})$ should contain some \kl{letter} in $A \setminus B$.
    Thus by $\word \not\in \hat{\val}(\top ((\sum_{i \in \range{0, n} \setminus \set{j}} x_i)^*)^{-} \top)$, we have $\word \in \hat{\val}(\term[2])$.
    Hence, this completes the proof of the first statement.

    For $\LANG_{n+1} \not\models \term[1] \le \term[2]$:
    Let $A = \set{\ell_{i} \mid i \in \range{0, n}}$ and let $\val \in \LANG_{A}$ be the \kl{valuation} s.t.\
    $\val(x_i) = \set{\ell_{i}}$ for each $i$.
    Then we have:
    \begin{align*}
        \hat{\val}(\term[1]) &= \set{\ell_{i} \mid i \in \range{0, n}}^*, &
        \hat{\val}(\term[2]) &= \bigcup_{j \in \range{0, n}} \set{\ell_i \mid i \in \range{0, n} \setminus \set{j}}^*.
    \end{align*}
    Hence by $\val \not\models \term[1] \le \term[2]$, this completes the proof of the second statement.
\end{proof}
Summarizing the above, for  $\KAtermclass_{\set{-}}$ \kl{terms}, we have:
\begin{align*}
    & \mathrm{EqT}(\LANG_{0}) \supsetneq \mathrm{EqT}(\LANG_{1}) \supsetneq \mathrm{EqT}(\LANG_{2}) \supsetneq \mathrm{EqT}(\LANG_{3}) \supsetneq \dots \\
    & \supsetneq \mathrm{EqT}(\LANG_{n}) \supsetneq \dots \supsetneq \mathrm{EqT}(\LANG_{\aleph_0}) = \mathrm{EqT}(\LANG).
\end{align*}
Here, $\mathrm{EqT}(\algclass)$ denotes the \kl{equational theory} of a class $\algclass$ for $\KAtermclass_{\set{-}}$ \kl{terms}.

\begin{rem}
    The \kl{equation} used in \Cref{theorem: separation gen} is based on the the following quantifier-free formula:
    \begin{align*}
        \LANG_{n} &\models ((\sum_{i \in \range{0, n}} x_i)^{*} = \top) \rightarrow (\bigvee_{j \in \range{0, n}} (\sum_{i \in \range{0, n} \setminus \set{j}} x_i)^{*} = \top),\\
        \LANG_{n + 1} &\not\models ((\sum_{i \in \range{0, n}} x_i)^{*} = \top) \rightarrow (\bigvee_{j \in \range{0, n}} (\sum_{i \in \range{0, n} \setminus \set{j}} x_i)^{*} = \top).
    \end{align*}
\end{rem}

\begin{rem}[open]\label{remark: open strictness}
    In the above construction, we need full complements.
    We leave open whether the hierarchy above is infinite for $\KAtermclass_{\set{\com{x}, \com{\id}}}$ (resp.\ $\KAtermclass_{\set{\com{x}}}$, $\KAtermclass_{\set{\com{\id}}}$, and $\KAtermclass$ with top) \kl{terms}.
\end{rem}
Note that, for some fragments, the hierarchy is collapsed, e.g., \Cref{corollary: bounded alphabet}, \Cref{theorem: KA collapse}, and \Cref{theorem: completeness word LANG2}.
In the next section, we show that the hierarchy is collapsed for \kl{words} with \kl{variable complements} (\Cref{theorem: completeness word LANG2}).

\section{Completeness theorem of the equational theory for the word fragment}\label{section: completeness word}
In this section, we show a completeness theorem for the \kl{equational theory} of $\LANG_{\alpha}$ for \kl{words} over $\tilde{\vsig}_{\com{\id}}$.
More precisely, we present a sound and complete equational proof system with a recursive set of axioms.
Notice that \kl{words-to-letters valuations} need an unbounded number of \kl{letters}, so they may not be compatible with $\LANG_{n}$ when $n$ is bounded.
In the following, we consider other \kl{valuations}.
Let $\mathcal{E}$ be a set of \kl{equations}.
We define $(=_{\mathcal{E}}) \subseteq \tilde{\vsig}_{\com{\id}}^{*} \times \tilde{\vsig}_{\com{\id}}^{*}$ as the minimal congruence (and equivalence) relation subsuming $\mathcal{E}$, i.e., the minimal relation satisfying the following:
\begin{itemize}
    \item $(=_{\mathcal{E}})$ is an equivalence relation: reflexive, symmetric, and transitive,
    \item $(=_{\mathcal{E}})$ is a congruence relation: if $\word[1] =_{\mathcal{E}} \word[2]$ and $\word[1]' =_{\mathcal{E}} \word[2]'$, then $\word[1] \word[1]' =_{\mathcal{E}} \word[2] \word[2]'$,
    \item if $(\word[1] = \word[2]) \in \mathcal{E}$, then $\word[1] =_{\mathcal{E}} \word[2]$.
\end{itemize}
We write $\mathcal{E} \vdash \word[1] = \word[2]$ if $\word[1] =_{\mathcal{E}} \word[2]$.

\subsection{On $\LANG_{0}$}
For a \kl{word} $\word = x_0 \dots x_{n-1} \in \tilde{\vsig}_{\com{\id}}^*$, we write $\mathsf{Occ}(\word)$ for the set $\set{x_i \mid i \in \range{0, n-1}}$.
\begin{thm}\label{theorem: completeness word LANG0}
    Let $\mathcal{E}_0 \defeq \set{xy = yx, xx = x, z\com{z} = \com{\id}, \com{\id}x = \com{\id} \mid x, y \in \tilde{\vsig}_{\com{\id}}, z \in \vsig}$.
    For all \kl{words} $\word[1], \word[2] \in \tilde{\vsig}_{\com{\id}}^*$, we have:
    \[\LANG_{0} \models \word[1] = \word[2] \quad\Leftrightarrow\quad \mathcal{E}_0 \vdash \word[1] = \word[2].\]
\end{thm}
\begin{proof}
    By \Cref{prop: LANG0 and Boolean algebra}, we have that $\LANG_{0} \models \word[1] = \word[2]$ iff $\word[1] = \word[2]$ is valid in Boolean algebra where
    the \kl{composition} ($\compo$) maps to the conjunction, the empty constant ($\id$) mapsto the true constant, and the complement ($\bl^{-}$) maps to the complement.
    Then $\mathcal{E}_0 \vdash \word[1] = \word[2]$ iff $\bigvee \left\{\begin{aligned}
        &\mathsf{Occ}(\word[1]) = \mathsf{Occ}(\word[2]),\\
        &\bigwedge \left\{
            \begin{aligned}
              &  \com{\id} \in \mathsf{Occ}(\word[1]) \lor (\exists z \in \vsig,\ \set{z, \com{z}} \subseteq \mathsf{Occ}(\word[1])),\\
              &  \com{\id} \in \mathsf{Occ}(\word[2]) \lor (\exists z \in \vsig,\ \set{z, \com{z}} \subseteq \mathsf{Occ}(\word[2]))
            \end{aligned}
        \right\}
    \end{aligned}\right\}$ iff $\word[1] = \word[2]$ is valid in Boolean algebra (the below case of the disjunction denotes that both the translated formulas in propositional logic are equivalent to the false constant).
\end{proof}

\subsection{On $\LANG_{1}$}
For a \kl{word} $\word[1] = x_0 \dots x_{n-1} \in \tilde{\vsig}_{\com{\id}}^{*}$ and $X \subseteq \tilde{\vsig}_{\com{\id}}$,
we write $\len{\word[1]}_{X}$ for the number $\#(\set{i \in \range{0, n-1} \mid x_{i} \in X})$.
Particularly, we write $\len{\word[1]}_{x}$ for $\len{\word[1]}_{\set{x}}$.
For a \kl{letter} $a$ and $n \in \nat$, we write $a^{n}$ for the \kl{word} $a \dots a$ of \kl{length} $n$.
\begin{thm}\label{theorem: completeness word LANG1}
    Let $\mathcal{E}_1 \defeq \set{xy = yx \mid x, y \in \tilde{\vsig}_{\com{\id}}}$.
    For all \kl{words} $\word[1], \word[2] \in \tilde{\vsig}_{\com{\id}}^*$, we have:
    \[\LANG_{1} \models \word[1] = \word[2] \quad \Leftrightarrow \quad \mathcal{E}_{1} \vdash \word[1] = \word[2].\]
\end{thm}  
\begin{proof}
    ($\Leftarrow$):
    Because the commutative law $xy = yx$ holds for all \kl{valuations} in $\LANG_{1}$.
    ($\Rightarrow$):
    It suffices to show that $\forall x \in \tilde{\vsig}_{\com{\id}}, \len{\word[1]}_{x} = \len{\word[2]}_{x}$.
    Assume that $\len{\word[1]}_{x} \neq \len{\word[2]}_{x}$ for some $x \in \tilde{\vsig}$.
    By flipping the sign of $x$, WLOG, we can assume that $x \in \vsig$.
    By swapping $\word[1]$ and $\word[2]$, WLOG, we can assume that $\len{\word[1]}_{x} < \len{\word[2]}_{x}$.
    Let $m \defeq 1 + \len{\word[1]}_{(\vsig \setminus \set{x}) \cup \set{\com{\id}}}$.
    Let $\val \in \LANG_{\set{\mathtt{a}}}$ be the \kl{valuation} defined by:
    \[\val(y) \ \defeq\  \begin{cases}
        \set{\mathtt{a}^{n} \mid n \ge m} & (y = x)   \\
        \set{\mathtt{a}^{n} \mid n \ge 1} & (y \neq x).
    \end{cases}\]
    Then,
    \begin{align*}
        \min\set{n \in \nat \mid \mathtt{a}^{n} \in \hat{\val}(\word[1])} &= m \len{\word[1]}_{x} + \len{\word[1]}_{(\vsig \setminus \set{x}) \cup \set{\com{\id}}} \\
        & < m (\len{\word[1]}_{x} + 1) \tag{By $\len{\word[1]}_{(\vsig \setminus \set{x}) \cup \set{\com{\id}}} < m$}\\
        &\le m \len{\word[2]}_{x} + \len{\word[2]}_{(\vsig \setminus \set{x}) \cup \set{\com{\id}}} \tag{By $\len{\word[1]}_{x} + 1 \le \len{\word[2]}_{x}$}\\
        &= \min\set{n \in \nat \mid \mathtt{a}^{n} \in \hat{\val}(\word[2])}.
    \end{align*}
    Thus $\hat{\val}(\word[1]) \setminus \hat{\val}(\word[2]) \neq \emptyset$, which contradicts $\LANG_{1} \models \word[1] = \word[2]$.
    Hence, $\forall x \in \tilde{\vsig}, \len{\word[1]}_{x} = \len{\word[2]}_{x}$.
    Next, assume that $\len{\word[1]}_{\com{\id}} \neq \len{\word[2]}_{\com{\id}}$.
    WLOG, we can assume that $\len{\word[1]}_{\com{\id}} < \len{\word[2]}_{\com{\id}}$.
    Let $\val \in \LANG_{\set{\mathtt{a}}}$ be the \kl{valuation} defined by: $\val(y) \defeq \set{\mathtt{a}^{n} \mid n \ge 1}$.
    Then we have $\hat{\val}(\word[1]) \setminus \hat{\val}(\word[2]) \neq \emptyset$ by the following:
    \begin{align*}
        \min\set{n \in \nat \mid \mathtt{a}^{n} \in \hat{\val}(\word[1])} &= \len{\word[1]}_{\com{\id}} + \len{\word[1]}_{\vsig}
        < \len{\word[2]}_{\com{\id}} + \len{\word[2]}_{\vsig} %
        = \min\set{n \in \nat \mid \mathtt{a}^{n} \in \hat{\val}(\word[2])}.
    \end{align*}
    Thus $\hat{\val}(\word[1]) \setminus \hat{\val}(\word[2]) \neq \emptyset$, which contradicts $\LANG_{1} \models \word[1] = \word[2]$.
    Hence, we have $\forall x \in \tilde{\vsig}_{\com{\id}}, \len{\word[1]}_{x} = \len{\word[2]}_{x}$.
    Therefore, $\mathcal{E}_{1} \vdash \word[1] = \word[2]$.
\end{proof}

\subsection{On $\LANG_{\alpha}$ where $\alpha \ge 2$}
What about for $\LANG_{2}$?
In the conference version, we have shown that the \kl{equational theory} coincides with the word equivalence \cite[Thm.\ 36]{nakamuraWordstoLettersValuationsLanguage2023} (\Cref{corollary: completeness word LANG2}) if the number of variables is at most one and the complement of the empty constant ($\com{\id}$) does not occur.
However, when $\com{\id}$ may occur, there are some non-trivial \kl{equations}, as follows.
\begin{ex}\label{example: LANG2}
    $\LANG \models \com{\id} z \com{z} \com{\id} = \com{\id} \com{z} z \com{\id}$ holds.
    Let $\val \in \LANG$.
    Note that $\eps \in \hat{\val}(z)$ or $\eps \in \hat{\val}(\com{z})$.
    W.r.t.\ ``$\val \models $'', if $\eps \in \hat{\val}(z)$, then by $\com{\id} z \le \com{\id}$, $z \com{\id} \le \com{\id}$ and $\id \le z$, we have:
    \begin{align*}
        \com{\id} z \com{z} \com{\id} \le \com{\id} \com{z} \com{\id} &\le \com{\id} \com{z} z \com{\id} \le \com{\id} \com{z} \com{\id} \le \com{\id} z \com{z} \com{\id}.
    \end{align*}
    We can show the case when $\eps \in \hat{\val}(\com{z})$ in the same way.
\end{ex}
Nevertheless, we have the following completeness theorem.
\begin{thm}\label{theorem: completeness word LANG2}
    Let $\alpha \ge 2$.    Let $\mathcal{E}_2$ be the set of the following \kl{equations}:
    \[\com{\id} z^{c_0} \com{z}^{d_0} \dots z^{c_{k-1}} \com{z}^{d_{k-1}} \com{\id} = \com{\id} \com{z}^{d_0} z^{c_0} \dots \com{z}^{d_{k-1}} z^{c_{k-1}} \com{\id}\]
    where $z \in \tilde{\vsig}$ and $k, c_0, d_0, \dots, c_{k-1}, d_{k-1}> 0$.
    For all \kl{words} $\word[1], \word[2] \in \tilde{\vsig}_{\com{\id}}^*$, we have:
   \[\LANG_{\alpha} \models \word[1] = \word[2] \quad\Leftrightarrow\quad \mathcal{E}_2 \vdash \word[1] = \word[2].\]
\end{thm}
We show \Cref{theorem: completeness word LANG2} in the following.

\subsubsection{Proof of the soundness (the direction ($\Leftarrow$) in \Cref{theorem: completeness word LANG2})}
For $z \in \tilde{\vsig}, k, c_0, d_0, \dots, c_{k-1}, d_{k-1}> 0$, we prove the following:
\[\LANG \models \com{\id} z^{c_0} \com{z}^{d_0} \dots z^{c_{k-1}} \com{z}^{d_{k-1}} \com{\id} = \com{\id} \com{z}^{d_0} z^{c_0} \dots \com{z}^{d_{k-1}} z^{c_{k-1}} \com{\id}.\]
Let $\val \in \LANG$.
Note that either $\eps \in \hat{\val}(z)$ or $\eps \in \hat{\val}(\com{z})$ holds.
Suppose $\eps \in \hat{\val}(z)$.
Note that, w.r.t.\ ``$\val \models $'', for $c, c', d > 0$, we have:
\begin{align*}
    (\id \union \com{z}) z^{c} \com{z}^{d} & \le (z \union \com{z}) \com{z}^{d} \tag{By $\term[3] \le \top = z \union \com{z}$}\\
    &\le (z^{c'} \union z^{c'} \com{z}) \com{z}^{d} = z^{c'} \com{z}^{d} (\id \union \com{z}). \tag{By $\id \le z$}
\end{align*}
Thus, we have:
\begin{align*}
    &\com{\id} z^{c_0} \com{z}^{d_0} z^{c_1} \com{z}^{d_1} \dots z^{c_{k-1}} \com{z}^{d_{k-1}} \com{\id}\\
    &\le \com{\id} \com{z}^{d_0} (\id \union \com{z}) z^{c_1} \com{z}^{d_1} \dots z^{c_{k-1}} \com{z}^{d_{k-1}} \com{\id} \tag{By $\com{\id} \term[3] \le \com{\id}$ and $\id \le \id \union \com{z}$}\\ 
    &\le \com{\id} \com{z}^{d_0} z^{c_0} \com{z}^{d_1} \dots \com{z}^{d_{k-2}} z^{c_{k-2}} \com{z}^{d_{k-1}} (\id \union \com{z}) \com{\id} \tag{By $(\id \union \com{z}) z^{c} \com{z}^{d} \le z^{c'} \com{z}^{d} (\id \union \com{z})$, iteratively}\\ 
    &\le \com{\id} \com{z}^{d_0} z^{c_0} \com{z}^{d_1} \dots \com{z}^{d_{k-2}} z^{c_{k-2}} \com{z}^{d_{k-1}} z^{c_{k-1}} \com{\id}. \tag{By $\term[3] \com{\id} \le \com{\id}$ and $\id \le z$}
\end{align*}
We can show the converse direction in the same way.
Therefore, we have obtained $\val \models \com{\id} z^{c_0} \com{z}^{d_0} \dots z^{c_{k-1}} \com{z}^{d_{k-1}} \com{\id} = \com{\id} \com{z}^{d_0} z^{c_0} \dots \com{z}^{d_{k-1}} z^{c_{k-1}} \com{\id}$.
We can show the case when $\id \in \hat{\val}(\com{x})$ in the same way.
Hence, this completes the proof.

\subsubsection{Proof of the completeness (the direction ($\Rightarrow$) in \Cref{theorem: completeness word LANG2})}
It suffices to prove that when $\alpha = 2$.
Note that by $\LANG_1 \models \word[1] = \word[2]$, we have $\forall z \in \tilde{\vsig}_{\com{\id}}, \len{\word[1]}_{z} = \len{\word[2]}_{z}$ (\Cref{theorem: completeness word LANG1}).

For each $z \in \tilde{\vsig}_{\com{\id}}$,
we say that $z$ is \intro*\kl{positive} if $z \in \vsig$
and $z$ is \intro*\kl{negative} if $z \in \tilde{\vsig}_{\com{\id}} \setminus \vsig$.
We prepare the following two lemmas:
\begin{lem}\label{lemma: negative}
    If $\LANG_2 \models \word[1] = \word[2]$, then the $i$-th \kl{negative} \kl{letters} occurring in $\word[1]$ and $\word[2]$ are the same \kl{letter}.
\end{lem}
\begin{proof}
    We prove the contraposition.
    Let $\com{x}, \com{y} \in \tilde{\vsig}_{\com{\id}} \setminus \vsig$ be the $i$-th ($1$-indexed) \kl{negative} \kl{letters} occurring in $\word[1]$ and $\word[2]$ such that $x \neq y$.
    WLOG, we can assume that $y \neq \id$.
    Let $c \defeq \len{\word[1]}_{\tilde{\vsig}_{\com{\id}} \setminus \vsig}$ (note that $c = \len{\word[2]}_{\tilde{\vsig}_{\com{\id}} \setminus \vsig}$).
    Let $A = \set{\mathtt{a}, \mathtt{b}}$ and let $\val \in \LANG_{A}$ be the \kl{valuation} defined by:
    \begin{align*}
        \val(z) \ \defeq\  \begin{cases}
            \set{\eps, \mathtt{a}} & \mbox{if $z = y$}\\
            \set{\eps, \mathtt{b}} & \mbox{otherwise}.
        \end{cases}
    \end{align*}
    Then there are $\word[1]' \in A^{i-1}$ and $\word[1]'' \in A^{c - i}$ s.t.\  $\word[1]'\mathtt{a}\word[1]'' \in \hat{\val}(\word[1])$,
    by $\eps \in \hat{\val}(z)$ for $z \in \vsig$, $A \cap \hat{\val}(\com{z}) \neq \emptyset$ for $z \in \vsig \cup \set{\id}$ and $\mathtt{a} \in \hat{\val}(\com{x})$.
    Next, assume that $\word[1]'\mathtt{a}\word[1]'' \in \hat{\val}(\word[2])$.
    Because the number of \kl{negative} \kl{letters} occurring in $\word[2]$ ($= c$) is equivalent to the \kl{length} of $\word[1]'\mathtt{a}\word[1]''$, and $\eps \not\in \hat{\val}(z)$ for $z \in \tilde{\vsig}_{\com{\id}} \setminus \vsig$, each \kl{negative} \kl{letter} should map to a \kl{letter}.
    However, because the $i$-th \kl{negative} \kl{letter} occurring in $\word[2]$ is $\com{y}$, we have $\mathtt{a} \not\in \hat{\val}(\com{y})$, thus reaching a contradiction.
    Thus, $\word[1]'\mathtt{a}\word[1]'' \not\in \hat{\val}(\word[2])$, and hence $\LANG_{2} \not \models \word[1] = \word[2]$.
\end{proof}

\begin{lem}\label{lemma: positive}
    If $\LANG_2 \models \word[1] = \word[2]$, then the following hold:
    \begin{itemize}
        \item The $i$-th and $(i+1)$-th \kl{positive} \kl{letters} occurring in $\word[1]$ are adjacent if and only if those in $\word[2]$ are adjacent.
        \item The first \kl{positive} \kl{letter} occurring in $\word[1]$ is the left-most if and only if that in $\word[2]$ is the left-most.
        \item The last \kl{positive} \kl{letter} occurring in $\word[1]$ is the right-most if and only if that in $\word[2]$ is the right-most.
    \end{itemize}
\end{lem}
\begin{proof}
    We only show the first statement (the remaining two can be shown by using the same \kl{valuation}).
    We prove the contraposition.
    WLOG, we can assume that the $i$-th and $(i+1)$-th \kl{positive} \kl{letters} occurring in $\word[1]$ are not adjacent and those in $\word[2]$ are adjacent.
    Let $c \defeq \len{\word[1]}_{\vsig}$ (note that $c = \len{\word[2]}_{\vsig}$).
    Let $A \defeq \set{\mathtt{a}, \mathtt{b}}$ and let $\val \in \LANG_{A}$ be the \kl{valuation} defined by:
    \begin{align*}
        \hat{\val}(z) &\ \defeq\  \ljump{(\mathtt{a} A^*) \cap (A^* \mathtt{a})}\\
        &\  (=\  \set{c_0 \dots c_{n-1} \in \set{\mathtt{a}, \mathtt{b}}^* \mid n \ge 1, c_0 = \mathtt{a}, c_{n-1} = \mathtt{a}}).
    \end{align*}
    Then there is a \kl{word} $\word[3] \in \ljump{(\mathtt{b}^* \mathtt{a})^{i-1} \mathtt{b}^* \mathtt{a} \mathtt{b}^{+} \mathtt{a} \mathtt{b}^{*} (\mathtt{a} \mathtt{b}^{*})^{c-i-1}} \cap \hat{\val}(\word[1])$,
    by $\mathtt{a} \in \hat{\val}(z)$ for $z \in \vsig$ and $\mathtt{b} \in \hat{\val}(\com{z})$ for $z \in \vsig \cup \set{\id}$.
    Note that the $i$-th and $(i+1)$-th \kl{positive} \kl{letters} occurring in $\word[1]$ are not adjacent; thus we can map the (non-empty) \kl{word} (over \kl{negative} \kl{letters}) between the two \kl{positive} \kl{letters} to some \kl{word} of the form $\mathtt{b}^{+}$.
    Next, assume that $\word[3] \in \hat{\val}(\word[2])$.
    Because the number of \kl{positive} \kl{letters} occurring in $\word[2]$ ($= c$) is equivalent to the number of $\mathtt{a}$ occurring in $\word[3]$,
    each \kl{positive} \kl{letter} should map to $\mathtt{a}$.
    However, because the $i$-th and $(i+1)$-th \kl{positive} \kl{letters} are adjacent, we have $\ljump{(\mathtt{b}^* \mathtt{a})^{i-1} \mathtt{b}^* \mathtt{a} \mathtt{b}^{+} \mathtt{a} \mathtt{b}^{*} (\mathtt{a} \mathtt{b}^{*})^{c-i-1}} \cap \hat{\val}(\word[2]) = \emptyset$.
    Thus, $\word[3] \not\in \hat{\val}(\word[2])$, and hence $\LANG_{2} \not \models \word[1] = \word[2]$.
\end{proof}

Now, we show the completeness theorem, using \Cref{lemma: negative,lemma: positive} with flipping signs.
For a \kl{word} $\word[3]$, we write $\word[3]_{\restriction n}$ for the prefix of $\word[3]$ of length $n$.
First, by \Cref{lemma: negative}, we have the following.
\begin{claim}\label{claim: 1'}
    For each $n$, there are two pairs $\tuple{\word[1]', \word[2]'}$ and $\tuple{\word[1]'', \word[2]''}$ of \kl{words} of the same \kl{length} such that
    \begin{itemize}
        \item $\word[1]_{\restriction n} = \word[1]' \word[1]''$ and $\word[2]_{\restriction n} = \word[2]' \word[2]''$,
        \item $\forall z \in \tilde{\vsig}_{\com{\id}}, \len{\word[1]'}_{z} = \len{\word[2]'}_{z}$,
        \item $\exists z_0 \in \vsig, \word[1]'', \word[2]'' \in \set{z_0, \com{z}_0}^*$.
    \end{itemize}
\end{claim}
\begin{proof}[Claim proof]
    By induction on $n$.
    Case $n = 0$: Trivial, by letting $\word[1]' = \word[2]' = \word[1]'' = \word[2]'' = \id$.
    Case $n > 0$:
    Let $\tuple{\word[1]', \word[2]'}$, $\tuple{\word[1]'', \word[2]''}$, and $z_0$ be the ones obtained by IH w.r.t.\ $n-1$.
    Let $x$ and $y$ be s.t.\ $\word[1]_{\restriction n} = \word[1]' \word[1]'' x$ and $\word[2]_{\restriction n} = \word[2]' \word[2]'' y$.
    We distinguish the following cases:
    \begin{itemize}
        \item Case $\len{\word[1]''}_{z_0} = \len{\word[2]''}_{z_0}$ and $\len{\word[1]''}_{\com{z}_0} = \len{\word[2]''}_{\com{z}_0}$:
        If $y \neq x$ and $y \neq \com{x}$,
        then by flipping the sign of $x$ and $y$, WLOG, we can assume that $x = \com{z}$ and $y = \com{z}'$ for some $z, z' \in \vsig \cup \set{\id}$ s.t.\ $z \neq z'$.
        However, this contradicts \Cref{lemma: negative}; note that $x$ and $y$ are the $i$-th \kl{negative} \kl{letter} occurring in $\word[1]$ and $\word[2]$ for some $i$, because $\len{\word[1]' \word[1]''}_{\tilde{\vsig}_{\com{\id}} \setminus \vsig} = \len{\word[2]' \word[2]''}_{\tilde{\vsig}_{\com{\id}} \setminus \vsig}$.
        Hence, $y = x$ or $y = \com{x}$ holds.
        Thus, the pair of $\tuple{\word[1]' \word[1]'', \word[2]' \word[2]''}$ and $\tuple{x, y}$ satisfy the condition.
        \item Otherwise:
        By $\len{\word[1]''} = \len{\word[2]''}$, we have either
        $(\len{\word[1]''}_{z_0} < \len{\word[2]''}_{z_0} \land \len{\word[1]''}_{\com{z}_0} > \len{\word[2]''}_{\com{z}_0})$ or
        $(\len{\word[1]''}_{z_0} > \len{\word[2]''}_{z_0} \land \len{\word[1]''}_{\com{z}_0} < \len{\word[2]''}_{\com{z}_0})$ holds.
        \begin{itemize}
            \item Case $y \not\in \set{z_0, \com{z}_0}$:
            By flipping the sign of $z_0$ and $x$, WLOG, we can assume that $\len{\word[1]''}_{\com{z}_0} < \len{\word[2]''}_{\com{z}_0}$
            and that $x = \com{z}$ for some $z \in (\vsig \cup \set{\id})$ s.t.\ $z \neq z_0$.
            However, this contradicts \Cref{lemma: negative}; note that $x$ and $\com{z}_0$ are the $i$-th \kl{negative} \kl{letter} occurring in $\word[1]$ and $\word[2]$ for some $i$, because
            $\len{\word[1]'}_{\tilde{\vsig}_{\com{\id}} \setminus \vsig} = \len{\word[2]'}_{\tilde{\vsig}_{\com{\id}} \setminus \vsig}$ and
            $\len{\word[1]''}_{\tilde{\vsig}_{\com{\id}} \setminus \vsig} = \len{\word[1]''}_{\com{z}_0} < \len{\word[2]''}_{\com{z}_0} = \len{\word[2]''}_{\tilde{\vsig}_{\com{\id}} \setminus \vsig}$.
            \item Case $x \not\in \set{z_0, \com{z}_0}$: Similarly to the above, we reach a contradiction.
            \item Otherwise:
            Since $x, y \in \set{z_0, \com{z}_0}$,
            the pair of $\tuple{\word[1]', \word[2]'}$ and $\tuple{ \word[1]'' x, \word[2]'' y}$ satisfies the condition.
        \end{itemize}
    \end{itemize}
    Hence, this completes the proof.
\end{proof}
As an immediate consequence of \Cref{claim: 1'}, there are $m \in \nat$, pairs $\tuple{\word[1]_i, \word[2]_i}$ of \kl{words} of the same non-zero length, and $z_i \in \tilde{\vsig}_{\com{\id}}\setminus \vsig$ (where $i \in \range{0, m-1}$) such that
\begin{itemize}
    \item $\word[1] = \word[1]_0 \dots \word[1]_{m-1}$ and $\word[2] = \word[2]_0 \dots \word[2]_{m-1}$,
    \item for each $i < m$, 
    \begin{itemize}
        \item if $z_i = \com{\id}$, then $\word[1]_i = \word[2]_i = \com{\id}^{n}$ for some $n \ge 1$,
        \item otherwise, $\word[1]_i, \word[2]_i \in \set{z_i, \com{z}_i}^{+}$, $\len{\word[1]_i}_{z_i} = \len{\word[2]_i}_{z_i}$, and $\len{\word[1]_i}_{\com{z}_i} = \len{\word[2]_i}_{\com{z}_i}$,
    \end{itemize}
    \item for each $i < m - 1$, we have $z_{i} \neq z_{i+1}$.
\end{itemize}

Moreover, by \Cref{lemma: positive}, each pair $\tuple{\word[1]_i, \word[2]_i}$ is of the following form.
\begin{claim}\label{lemma: positive'}
    For each $i$,
    $\word[1]_{i} = \word[2]_{i}$ holds or the following all hold:
    \begin{itemize}
        \item $z_i \neq \com{\id}$ and $z_{i-1} = z_{i+1} = \com{\id}$,
        \item $i \neq 0$ and $i \neq m-1$,
        \item there are $z \in \set{z_i, \com{z}_i}$, $k > 0$, and $c_{0}, d_{0}, \dots, c_{k-1}, d_{k-1} > 0$ s.t.\ 
        \begin{align*}
            \word[1]_{i} &= z^{c_0} \com{z}^{d_0} \dots z^{c_{k-1}} \com{z}^{d_{k-1}},\\
            \word[2]_{i} &= \com{z}^{d_0} z^{c_0} \dots \com{z}^{d_{k-1}} z^{c_{k-1}}.
        \end{align*}
    \end{itemize}
\end{claim}
\begin{proof}[Claim proof]
    Let $z = \com{z}_i$.
    If $\com{z} {(= z_i)} = \com{\id}$, then $\word[1]_{i} = \word[2]_{i}$.
    Otherwise, let $\word[1]_{i} = x_{0} \dots x_{n-1}$ and $\word[2]_{i} = y_{0} \dots y_{n-1}$ where $n > 0$ and $x_0, y_0, \dots, x_{n-1}, y_{n-1} \in \set{z, \com{z}}$.
    Note that $\len{\word[1]_{i}}_{z} = \len{\word[2]_{i}}_{z}$ and $\len{\word[1]_{i}}_{\com{z}} = \len{\word[2]_{i}}_{\com{z}}$.
    We distinguish the following cases:
    \begin{itemize}
        \item Case $x_{0} = y_{0}$:
        If $n = 0$, then $\word[1]_{i} = \word[2]_{i}$.
        Otherwise, $n \ge 1$.
        Assume $x_{1} \neq y_{1}$.
        Then $x_{0} = x_{1}$ or $y_{0} = y_{1}$ holds.
        By flipping the sign of $z$ and by swapping $\word[1]_i$ and $\word[2]_i$, WLOG, we can assume that $x_{0} = x_{1} = z$.
        Let $j > 1$ be the minimal number s.t.\ $y_{j} = z$ (such $j$ exists by $\len{\word[1]_{i}}_{z} = \len{\word[2]_{i}}_{z}$).
        Then this contradicts \Cref{lemma: positive}, because $x_0$ and $x_1$ are adjacent, but $y_0$ and $y_1$ are not.
        Thus, $x_{1} = y_{1}$.
        Using the same argument iteratively, we have $x_{j} = y_{j}$ for each $j$.
        Hence, $\word[1]_{i} = \word[2]_{i}$.
        \item Case $x_{0} \neq y_{0}$:
        By flipping the sign of $z$, WLOG, we can assume that $x_{0} = z$.
        Then, because $\len{\word[1]_{i}}_{\com{z}} \ge 1$ and $\len{\word[2]_{i}}_{z} \ge 1$,
        the \kl{words} $\word[1]_i$ and $\word[2]_i$ are of the following form where $c_0, d_0 > 0$:
        \begin{align*}
            \word[1]_i &= z^{c_0} \com{z} \word[1]_i', &\word[2]_i &= \com{z}^{d_0} z \word[2]_i'.
        \end{align*}
        By \Cref{lemma: positive}, moreover, $\word[1]_i$ and $\word[2]_i$ are of the following form:
        \begin{align*}
            \word[1]_i &= z^{c_0} \com{z}^{d_0} \word[1]_i'', & \word[2]_i &= \com{z}^{d_0} z^{c_0} \word[2]_i''.
        \end{align*}
        By applying the same argument iteratively, $\word[1]_{i}$ and $\word[2]_{i}$ of the form in this claim.
        The remaining part shows some additional conditions.
        If $i = 0$ (resp.\ $i = m - 1$), then this contradicts \Cref{lemma: positive}, as $x_0 \neq y_0$ (resp.\ $x_{m-1} \neq y_{m-1}$).
        If $z_{i-1} \neq \com{\id}$, then by flipping the sign of $z_{i-1}$ (note that $z_{i-1} \neq z_{i}$), WLOG, we can assume that the right-most \kl{variable} in $\word[1]_{i-1}$ is \kl{positive}.
        Let $j$ be the number such that the $j$-th \kl{positive} occurrence in $\word[1]$ is the \kl{letter} $x_0$ ($= z$).
        Then the $(j-1)$-th and the $j$-th \kl{positive} occurrences are adjacent in $\word[1]$ but not adjacent in $\word[2]$, and thus this contradicts \Cref{lemma: positive}.
        Hence, $z_{i-1} = \com{\id}$.
        By the same argument, we also have $z_{i+1} = \com{\id}$.
        Hence, this completes the proof.
    \end{itemize}
\end{proof}
By \Cref{lemma: positive'}, if $\word[1]_i \neq \word[2]_i$, then $\word[1]_i$ and $\word[2]_i$ are occurs in $\word[1]$ and $\word[2]$ as follows:
\begin{align*}
    \word[1] &= \dots \com{\id}z^{c_0} \com{z}^{d_0} \dots z^{c_{k-1}} \com{z}^{d_{k-1}}\com{\id} \dots,\\
    \word[2] &= \dots \com{\id}\com{z}^{d_0} z^{c_0} \dots \com{z}^{d_{k-1}} z^{c_{k-1}}\com{\id} \dots.
\end{align*}
where $\word[1]_i = z^{c_0} \com{z}^{d_0} \dots z^{c_{k-1}} \com{z}^{d_{k-1}}$ and $\word[2]_i = \com{z}^{d_0} z^{c_0} \dots \com{z}^{d_{k-1}} z^{c_{k-1}}$.
Hence, $\mathcal{E}_{2} \vdash \word[1] = \word[2]$.

\subsection{Remarks}
By the results in this section, for \kl{words} over $\tilde{\vsig}_{\com{\id}}$, we have:
\begin{align*}
    & \mathrm{EqT}(\LANG_{0}) \supsetneq \mathrm{EqT}(\LANG_{1}) \supsetneq \mathrm{EqT}(\LANG_{2}) = \mathrm{EqT}(\LANG_{3}) = \dots \\
    & = \mathrm{EqT}(\LANG_{n}) = \dots = \mathrm{EqT}(\LANG_{\aleph_0}) = \mathrm{EqT}(\LANG).
\end{align*}
Here, $\mathrm{EqT}(\algclass)$ denotes the \kl{equational theory} of a class $\algclass$ for \kl{words} over $\tilde{\vsig}_{\com{\id}}$.

Additionally, as an immediate consequence of \Cref{theorem: completeness word LANG2}, we have that if $\com{\id}$ does not occur, the \kl{equational theory} coincides with the \kl{word} equivalence.
\begin{cor}\label{corollary: completeness word LANG2}
    Let $\alpha \ge 2$.
    For all \kl{words} $\word[1], \word[2] \in \tilde{\vsig}^*$, we have:
   \[\LANG_{\alpha} \models \word[1] = \word[2] \quad\Leftrightarrow\quad \emptyset \vdash \word[1] = \word[2].\]
\end{cor}
\begin{proof}
    By \Cref{theorem: completeness word LANG2}, as all the \kl{equations} in $\mathcal{E}_2$ contains $\com{\id}$.
\end{proof}
\Cref{corollary: completeness word LANG2} strengthens \cite[Thm.\ 36]{nakamuraWordstoLettersValuationsLanguage2023} from one variable \kl{words} to many variables \kl{words}, which settles an open question given in \cite[p.\ 198]{nakamuraWordstoLettersValuationsLanguage2023}.
\begin{rem}
    Since $\ljump{\word[1]}_{\tilde{\vsig}} = \set{\word[1]}$, \Cref{corollary: completeness word LANG2} implies that for all \kl{words} $\word[1], \word[2]$ over $\tilde{\vsig}$,
    \[\LANG \models \word[1] = \word[2] \quad\Leftrightarrow\quad \ljump{\word[1]}_{_{\tilde{\vsig}}} = \ljump{\word[2]}_{_{\tilde{\vsig}}}.\]
    However, for general terms, the direction $\Rightarrow$ fails.
    For example, when $x \neq y$,
    \begin{align*}
        \LANG & \models x \union \com{x} =  y \union \com{y}, & \ljump{x \union \com{x}}_{_{\tilde{\vsig}}} \neq \ljump{y \union \com{y}}_{_{\tilde{\vsig}}}.
    \end{align*}
    Thus, we need more axioms to characterize the \kl{equational theory}.
\end{rem} 
\section{Conclusion and future work}\label{section: conclusion}
We have introduced \kl{words-to-letters valuations}.
By using them, we have shown the decidability and complexity of the identity/variable/word inclusion problems (\Cref{corollary: identity decidable,corollary: variable decidable,corollary: word decidable}) and the universality problem (\Cref{corollary: universality decidable}) of the \kl{equational theory w.r.t.\ languages} for $\KAtermclass_{\set{\com{x}, \com{\id}}}$ \kl{terms}; in particular, the \kl[equational theory]{inequational theory} $\term[1] \le \term[2]$ is coNP-complete when $\term[1]$ does not contain Kleene-star (\Cref{corollary: star-free}).
We summarize the complexity result in \Cref{table: comparison}. We leave open the (finite) axiomatizability of the \kl{equational theory} of $\LANG$.
\begin{table}[ht]
    \centering
    \begin{tabular}{l||c|c||c}
        \hline
                                                                                                & \multicolumn{2}{c||}{$\LANG \models \term[1] \le \term[2]$} & {$\ljump{\term[1]} \subseteq \ljump{\term[2]}$ where $\vsig$ finite}                                                                                                                          \\
        \cline{2-3}\cline{4-4}
                                                                                                & complexity                                        & $\mathrm{l}(\term)$                                                                              & {complexity}\\
                                                                                                && (\Cref{corollary: bounded alphabet}) & (\hspace{-.05em}\cite{Meyer1973}\hspace{-.05em}\cite[Thm.\ 2.6]{Hunt1976})                                 \\
        \hline
        $\term[1] = \id$                                                                        & \textcolor{black}{coNP}-c (\Cref{corollary: identity decidable})                          & $0$                                                                          & {in \textcolor{black}{P}}      \\
        \hline
        $\term[1] = x$ ($x \in \vsig$)            & \textcolor{black}{coNP}-c (\Cref{corollary: variable decidable})                         & $1$                                                                    & in P        \\
        $\term[1] = \com{x}$ ($x \in \vsig$)            & \textcolor{black}{coNP}-c (\Cref{corollary: variable decidable})                         & $1$                                                                    & {\textcolor{black}{PSPACE}-c}                 \\
        $\term = \top$                                                                      & \textcolor{black}{coNP}-c (\Cref{corollary: universality decidable})                         & $1$                                                                    & {\textcolor{black}{PSPACE}-c} \\
        \hline
        $\term = \word$ ($\word \in \tilde{\vsig}^*$)  & \textcolor{black}{coNP}-c (\Cref{corollary: word decidable})                         & $\le \len{\word}$                  & {\textcolor{black}{PSPACE}-c}                 \\
        $\term$ is $\bl^{*}$-free                 & \textcolor{black}{coNP}-c  (\Cref{corollary: star-free})        & $\le \len{\term}$              & {\textcolor{black}{PSPACE}-c}                 \\
        \hline
        (unrestricted)                                                                     & PSPACE-c \cite{nakamuraFiniteRelationalSemantics2025}                                & $\omega$ & {\textcolor{black}{PSPACE}-c}                 \\
        \hline
    \end{tabular}
    \caption{Comparison between $\LANG$ and the standard language \kl{valuation} $\ljump{\bl}$ for $\KAtermclass_{\set{\com{x}, \com{\id}}}$.}
    \label{table: comparison}
\end{table}

Moreover, we have considered the \kl{equational theories} of $\LANG_{n}$ (where $n$ is bounded) and have shown that the hierarchy is infinite for $\KAtermclass_{\set{-}}$ \kl{terms} (\Cref{theorem: separation gen}).
We leave it open for $\KAtermclass_{\set{\com{x}, \com{\id}}}$ \kl{terms} and its some fragments (\Cref{remark: open strictness}).
Additionally, we have proved the completeness theorem for the \kl{word} fragment of $\KAtermclass_{\set{\com{x}, \com{\id}}}$ \kl{terms} w.r.t.\ languages (\Cref{theorem: completeness word LANG2}); as a corollary, the hierarchy is collapsed for the \kl{word} fragment of $\KAtermclass_{\set{\com{x}, \com{\id}}}$ \kl{terms}.
We also leave open the decidability/complexity and the (finite) axiomatizability of the \kl{equational theory} of $\LANG_{n}$ (cf.\ \Cref{table: comparison}).
 
\section*{Acknowledgements}
This work was supported by JSPS KAKENHI Grant Number JP21K13828 and JST ACT-X Grant Number JPMJAX210B, Japan.
\bibliographystyle{ws-ijfcs}
\bibliography{main}

\end{document}